%% file: main.tex
\documentclass[11pt,letterpaper]{article}
\usepackage{shortcuts/style}
\usepackage{shortcuts/shortcuts}
\usepackage{authblk}
\RenewDocumentCommand{\leanok}{ O{\lastlabelname} }{}
\newcommand*\samethanks[1][\value{footnote}]{\footnotemark[#1]}

\title{Temporal Conductance and Bounds on the Voter Model for Dynamic Networks}
\author[1]{Tatiana Rocha Avila\thanks{Supported by DFG Research Unit ADYN under grant DFG 411362735.}\thanks{We thank Kitty Meeks and Jessica Enright for helpful initial discussions which formed much of the basis of the definition of temporal conductance; these discussions took place during visits that were funded in part by EPSRC grants EP/T004878/1 and EP/V032305/1.}}
\author[2]{Holger Dell\samethanks[1]\samethanks}
\author[3]{John Lapinskas\samethanks}
\affil[1]{Goethe University Frankfurt, Germany}
\affil[2]{IT University of Copenhagen, Denmark}
\affil[3]{University of Bristol, UK}

\begin{document}
\maketitle
\input{sections/0_abstract}
\pagebreak
\input{sections/1_introduction}
\input{sections/2_preliminaries}
\input{sections/3_main_result}
\input{sections/4_lower_bound}

\input{sections/5_ai}
\bibliographystyle{plainurl}
\bibliography{main.bib}

\end{document}

%% file: sections/0_abstract.tex
\begin{abstract}
 The voter model is a classical stochastic process that models how opinions might spread through a network: at each step, every node lazily adopts the opinion of a random neighbour; eventually all nodes share the same opinion (\emph{consensus}).
  Stronger connectivity should yield faster consensus.
  Berenbrink, Giakkoupis, Kermarrec, and Mallmann-Trenn (ICALP 2016) make this precise via the network's conductance: if the network has $m$ edges, minimum degree $d_{\min}$, and conductance at least $\phi$, then the voter model reaches consensus in expected $O(m/(d_{\min}\phi))$ steps. Their results extend to dynamic networks with fixed vertex degrees by considering the network's conductance at each time step.

  We introduce \emph{temporal conductance}~$\Phi$, a more general connectivity measure for dynamic networks.
  Unlike static conductance, which collapses to~$0$ whenever some snapshot is disconnected, $\Phi$ captures connectivity through edges that appear at different times.
  We generalise the results of Berenbrink, Giakkoupis, Kermarrec, and Mallmann-Trenn\ from static conductance to temporal conductance, showing that the expected consensus time of the standard voter model is at most $O(m/(d_{\min}\Phi))$. Moreover, we prove that this bound is tight up to constant factors. We provide a full Lean formalisation of our results.
  We expect temporal conductance to be a useful primitive for analysing other dynamics on temporal networks, and potentially time-inhomogeneous Markov chains more generally.
\end{abstract}
\clearpage

%% file: sections/1_introduction.tex
\section{Introduction}

\begin{figure}[t]
  \centering
  \begin{tikzpicture}[
      x=1cm,
      y=1cm,
      >=Latex,
      node/.style={circle, draw=black, thick, minimum size=4.8mm, inner sep=0pt},
      opinA/.style={node, fill=blue!65!cyan},
      opinB/.style={node, fill=orange!85!yellow},
      undir/.style={black!30, line width=0.6pt},
      choice/.style={->, thick, black},
      lbl/.style={font=\scriptsize}
    ]

    \begin{scope}[shift={(-4.15,0)}]

      \begin{scope}[shift={(0,0)}]
        \node[opinB] (a) at (-1.0,  0.55) {};
        \node[opinA] (b) at ( 0.0,  0.95) {};
        \node[opinB] (c) at ( 1.0,  0.55) {};
        \node[opinA] (d) at ( 1.0, -0.55) {};
        \node[opinB] (e) at ( 0.0, -0.95) {};
        \node[opinA] (f) at (-1.0, -0.55) {};

        \draw[undir] (a)--(b)--(c)--(d)--(e)--(f)--(a);
        \draw[undir] (b)--(e);

        \draw[choice] (a) -- (f);
        \draw[choice] (b) -- (a);
        \draw[choice] (c) -- (b);
        \draw[choice] (d) -- (c);
        \draw[choice] (e) -- (f);
        \draw[choice] (f) to[out=160,in=200,looseness=8] (f);

        \node[lbl] at (0,-1.5) {chose neighbor to copy};
      \end{scope}

      \draw[->, very thick] (2.05,0) -- (3.15,0);
      \node[lbl] at (2.6,0.23) {update};

      \begin{scope}[shift={(5.2,0)}]
        \node[opinA] (a2) at (-1.0,  0.55) {};
        \node[opinB] (b2) at ( 0.0,  0.95) {};
        \node[opinA] (c2) at ( 1.0,  0.55) {};
        \node[opinB] (d2) at ( 1.0, -0.55) {};
        \node[opinA] (e2) at ( 0.0, -0.95) {};
        \node[opinA] (f2) at (-1.0, -0.55) {};

        \draw[undir] (a2)--(b2)--(c2)--(d2)--(e2)--(f2)--(a2);
        \draw[undir] (b2)--(e2);

        \node[lbl] at (0,-1.5) {after update};
      \end{scope}

    \end{scope}

  \end{tikzpicture}
  \caption{\label{fig:lazy-synchronous-voter}An update step in the standard voter model with two opinions (that is, all vertices act independently and synchronously). In the left panel, each vertex points to a neighbour that it chooses randomly; a self-loop means that the vertex is ``lazy'' and keeps its current opinion; this happens with probability~$\tfrac12$. In the right panel, every vertex has adopted the opinion of the vertex that it pointed to.}
\end{figure}
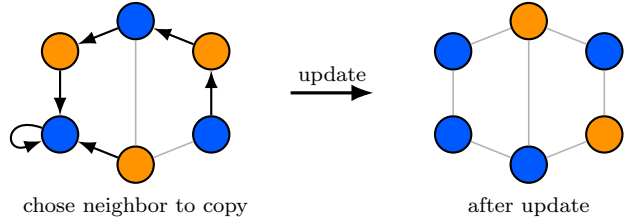

The \emph{voter model}~\cite{CliffordS73,HolleyL75} is a model of how opinions spread through a social network, see \cref{fig:lazy-synchronous-voter}.
It is most interesting not because it is realistic---it is too simplistic for that---but because it resists rigorous mathematical analysis despite its simplicity.
If we cannot analyze the voter model, then we are unlikely to be able to analyze realistic models of opinion dynamics, which are more complex.
The voter model serves as a testbed for methods in network dynamics more generally, and has been the subject of a large body of work in the mathematical physics~\cite{CastellanoFL09} and probability~\cite{Liggett1999} communities.

\paragraph{Consensus time.}
One of the most fundamental questions about the voter model is this: How long does it take for the process to reach consensus, that is, for all vertices to have the same opinion?
Results span a wide range of graph families: on the complete graph~$K_n$ the expected consensus time is $\Theta(n)$~\cite{Liggett1999}; on random regular graphs it is also $\Theta(n)$~\cite{DBLP:journals/siamdm/CooperFR09}; on heterogeneous (configuration-model) graphs it depends on the degree distribution~\cite{SoodRedner2005}; and on subcritical scale-free graphs the exponent varies in a rich phase diagram~\cite{FernleyOrtgiese2023}.

An important result by Berenbrink, Giakkoupis, Kermarrec, and Mallmann-Trenn~\cite{DBLP:conf/icalp/BerenbrinkGKM16} shows that the expected consensus time of the voter model is captured by the \emph{conductance} of the underlying graph, which is a measure of how well-connected the graph is. In particular, they show that the consensus time is at most $O(m/(d_{\min}\phi))$, where~$m$ is the number of edges in the graph, $d_{\min}$ is the minimum degree of any vertex, and~$\phi$ is the conductance of the graph. They also show that there are graphs where this upper bound is asymptotically tight.

\paragraph{Voter model in temporal graphs.}
When it comes to real-world social networks, the underlying graph is not static; it changes over time. For example, people make new friends, they interact with different subsets of people at different times, and during pandemics, lockdowns may be imposed or lifted. This motivates the study of the voter model as a test-bed for network dynamics on \emph{temporal graphs $\calG$}, which are sequences of graphs that change over time: $\calG=(G_t\colon t\in\N)$. In the context of spreading processes like the voter model, this is particularly useful in settings where the graph structure is independent from the spreading process being studied and can be observed in detail, such as epidemic modelling in livestock markets~\cite{ansari2021temporal} or schools~\cite{st2024nonlinear}.

The results of~\cite{DBLP:conf/icalp/BerenbrinkGKM16} are stated for the temporal setting as well, in the setting where vertex degrees are fixed with respect to time, by considering the conductance $\phi_t$ of the graph at each time step~$t$. They show that if $\sum_{t=0}^T \phi_t \ge Cm/d_{\min}$ for an appropriate constant $C$, then with probability at least $1/2$ the absorption time is at most $T$; in particular, if $\phi_t \ge \phi$ for all $t \ge 0$ then this generalises their result in the static case. However, both bounds are proved in an adaptive setting, where an adversary can choose each graph $G_t$ with full knowledge of the state of the voter model at time $t-1$, rather than the non-adaptive paradigm described above where the graph and voter model evolve independently. While their upper bound carries over to this setting immediately, their lower bound does not.

As a simple example, consider an $n$-vertex temporal graph consisting of one uniformly-random perfect matching per time step, drawn independently. It is not hard to show that such a graph has $O(n)$ expected absorption time in the voter model, but it is disconnected at every time step so $\sum_{t=0}^T \phi_t = 0$ for all $T$ --- that is, the results of~\cite{DBLP:conf/icalp/BerenbrinkGKM16} do not give any upper bound on the expected absorption time.

\paragraph{Our contribution.}
In this paper, we study the voter model on temporal graphs with fixed vertex degrees in the non-adaptive setting. We define a novel parameter~$\Phi(\calG)$ of \emph{temporal conductance} which measures the connectivity of such graphs.
We generalize the results of~\cite{DBLP:conf/icalp/BerenbrinkGKM16} and show that with probability at least $1/2$ the consensus time is at most $O(m/(d_{\min}\Phi(\calG)))$, and that there are temporal graphs for which this upper bound is asymptotically tight. Further, we generalise the stepwise bound of~\cite{DBLP:conf/icalp/BerenbrinkGKM16} to an interval-wise bound --- if time interval $I_i$ contains at least one time $t$ such that $G_t$ has conductance at least $\phi_i$, and $\sum_{i=0}^I \phi_i \ge Cm/d_{\min}$ for some suitable constant $C$, then we prove absorption with probability $1/2$ after $I$ intervals have passed. We prove this bound is asymptotically tight even if instead of requiring \emph{maximum} conductance at least $\phi_i$ over each interval $I_i$, we require \emph{average} conductance at least $\phi_i$.

In order to prove our results, we perform a drift analysis with a potential function inspired by~\cite{DBLP:conf/icalp/BerenbrinkGKM16}.
Dealing with temporal conductance turns out to pose significant challenges, which we overcome by splitting time into \emph{stable intervals} (during which the volume of each opinion class remains roughly constant, allowing conductance to be exploited) and
\emph{unstable intervals} (during which the volume of the minority opinion changes significantly, causing rapid decline of the potential even if conductance is poor).

After writing and checking the original proofs, we used Claude Code~\cite{claudecode} to formalize all of our proofs in the Lean~4 proof assistant \cite{lean4} on top of mathlib \texttt{v4.31.0} \cite{mathlib2020} using $\sim$38k lines of code.
We have verified the statements of \cref{thm:upper-bound-temporal-conductance,thm:upper-bound-window-conductance,thm:lower-bound-3regular,thm:lower-bound-odd-regular-averaged} and the requisite definitions. The proofs rely only on the standard axioms \texttt{propext}, \texttt{Classical.choice}, and \texttt{Quot.sound}, and the source code of the formalisation is accessible \href{https://holgerdell.github.io/temporal-voter}{here}.

\subsection{Related work}

\paragraph{Conductance in static graphs.} The \emph{conductance} of a Markov chain was defined by Jerrum and Sinclair~\cite{DBLP:conf/stoc/JerrumS88} in 1988 to prove their seminal result: counting the number of perfect matchings in bipartite graphs admits a randomized polynomial-time approximation scheme.
In order to do this, they proved that the conductance is closely tied to the Markov chain's \emph{mixing time}, and used this to show that randomly switching between perfect and near-perfect matchings in bipartite graphs is enough to approximately uniformly sample perfect matchings in polynomial time.
Sinclair and Jerrum~\cite{DBLP:journals/iandc/SinclairJ89} then developed this connection into a general framework for approximate counting and uniform generation, now central to the analysis of MCMC algorithms~\cite{levin2017markov}.
Conductance is also equivalent (up to squaring) to the spectral gap of the random walk matrix---a discrete Cheeger inequality established by Alon and Milman~\cite{DBLP:journals/jct/AlonM85} and Mihail~\cite{DBLP:conf/focs/Mihail89}---linking it to expander graph theory, spectral graph partitioning~\cite{DBLP:journals/jacm/AroraRV09}, and the mixing of reversible Markov chains more broadly~\cite{lawler1988bounds}.

\paragraph{Spreading and random walks on dynamic graphs.}
Random walks and spreading processes on dynamic graphs have been studied in various models.
Clementi, Macci, Monti, Pasquale, and Silvestri~\cite{DBLP:journals/siamdm/ClementiMMPS10} studied flooding time in the \emph{edge-Markovian} model---where each edge appears and disappears independently as a two-state Markov chain---and proved near-tight bounds.
Giakkoupis, Sauerwald, and Stauffer~\cite{DBLP:conf/icalp/GiakkoupisSS14} studied push-pull rumor spreading under adversarial rewiring, allowing degrees to vary by at most a constant factor, and showed that the protocol broadcasts to all vertices whenever the cumulative conductance satisfies $\sum_i \phi_i = \Omega(\log n)$, where~$\phi_i$ is the conductance at round~$i$.
Avin, Kouck{\'y}, and Lotker~\cite{DBLP:journals/rsa/AvinKL18} showed that an oblivious adversary can force exponential cover time on a simple random walk even when every snapshot is well-connected, but that the lazy random walk achieves polynomial cover and mixing time against any oblivious adversary.
Sauerwald and Zanetti~\cite{DBLP:conf/icalp/SauerwaldZ19} proved $O(n^2)$ mixing and hitting time bounds for random walks on any sequence of $d$-regular connected graphs via a local expansion argument.
Recently, Galanis, Goldberg, and Mifsud~\cite{GalanisGM26} studied random walks on dynamic graphs where edges evolve according to Glauber dynamics for the random-cluster model, and prove $O(\log n)$ mixing time in the subcritical regime, matching that of a static random regular graph.

\paragraph{Absorption time of the voter model.} It is well-known (see e.g.~\cite{liggett1985interacting}) that the voter model on static graphs is dual to a random walk model --- the $n$-opinion case on an $n$-vertex graph is dual to $n$ independent random walks, one starting on each vertex, with distinct walks coalescing into a single walk on collision. The absorption time of the voter model then corresponds to the time at which only one walk remains. As such, upper bounds on absorption time are often proved in terms of random walks and their properties. For example, Hassin and Peleg~\cite{hassin2001distributed} prove that in static graphs the absorption time is $O(t_{\textrm{meet}}\log n)$, where $t_{\textrm{meet}}$ is the expected time for two independent random walks to meet on the base graph; this was recently~\cite{10.1145/3576900} sharpened to $O(t_{\textrm{meet}}(1+\sqrt{t_{\textrm{mix}}/t_{\textrm{meet}}}\log n))$, where $t_{\textrm{mix}}$ is the mixing time of a single random walk. Meanwhile, Cooper, Els{\"{a}}sser, Ono, and Radzik~\cite{DBLP:journals/siamdm/CooperEOR13} bound the expected absorption time in terms of the second-largest eigenvalue of the random walk's transition matrix and its stationary distribution. None of these bounds are directly comparable to the $O(m/(d_{\min} \phi))$ bound from~\cite{DBLP:conf/icalp/BerenbrinkGKM16}; see their introduction for a more detailed discussion.

There is one other general upper bound on absorption in terms of conductance. In~\cite{DBLP:conf/icalp/BerenbrinkGKM16}, in addition to their main result of expected absorption time $O(m/(d_{\min}\phi))$ (the one stated in their abstract), Berenbrink, Giakkoupis, Kermarrec, and Mallmann-Trenn\ also prove expected absorption time $O(n\log n/\phi^2)$. This is a strictly weaker result for the important case of regular graphs, but can be stronger for graphs which are far from regular and have high conductance. We do not generalise this result to the non-adaptive setting, and consider it an interesting open question --- in~\cite{DBLP:conf/icalp/BerenbrinkGKM16} the result falls out relatively easily from their main proof, but in our setting this approach breaks.

\paragraph{Other notions of temporal conductance.}
DiTursi, Ghosh, and Bogdanov~\cite{DBLP:conf/icdm/DiTursiGB17} define a temporal conductance for community detection in dynamic networks by aggregating cut and volume over a time interval with a normalisation factor; their definition is motivated by data-mining considerations and differs from ours.
Berenbrink, Giakkoupis, Kermarrec, and Mallmann-Trenn~\cite{DBLP:conf/icalp/BerenbrinkGKM16} and Giakkoupis, Sauerwald, and Stauffer~\cite{DBLP:conf/icalp/GiakkoupisSS14} use per-round conductances as input parameters, with bounds in terms of their minimum or cumulative sum.
Our~$\Phi(\calG)$ instead optimises the window length automatically, yielding a single intrinsic parameter of the temporal graph.

\subsection{The voter model on temporal graphs}
As discussed earlier, the dynamic networks we consider in this paper are \emph{temporal graphs with fixed degrees}, which are arbitrary sequences of graphs on a common vertex set, with the only restriction being that the degree of each vertex does not depend on time.
The fixed degree assumption makes stochastic processes on the temporal graph much more well-behaved; for example, the lazy random walk becomes \emph{strongly ergodic} and has a unique stationary distribution.
The same assumption was also made by Berenbrink, Giakkoupis, Kermarrec, and Mallmann-Trenn~\cite{DBLP:conf/icalp/BerenbrinkGKM16}, because without it, the expected consensus time can be exponential in the number of vertices~\cite[Observation~1]{DBLP:conf/icalp/BerenbrinkGKM16}.
We give a formal definition next.
\begin{definition}[Temporal graph with fixed degrees]\label{def:temporal-graph} %
  Let $\calG$ be a sequence $(G_t\colon t\geq0)$ of undirected graphs on a common vertex set $V(\calG)$, where $G_t$ is the graph at time~$t$. We say $\calG$ is a \emph{temporal graph with fixed degrees} if the degree of each vertex $v\in V(\calG)$ in $G_t$ does not depend on~$t$. We denote this common degree by $d(v)$, and the (time-independent) number of edges of~$\calG$ by $m$, so that $m=\tfrac12\sum_{v\in V(\calG)} d(v)$ holds.
  For convenience, we assume $d(v)\geq1$ for all $v\in V(\calG)$, that is, there are no isolated vertices in $\calG$.
\end{definition}
We remark that the degree of each vertex is \emph{fixed} across time, but not necessarily \emph{constant}, that is, the degree of a vertex may depend on the number~$n$ of vertices but not on~$t$.

In the voter model, each vertex holds an opinion and repeatedly copies the opinion of a randomly chosen neighbour in the current graph; the temporal graph thus governs which interactions are available at each time step.

\begin{definition}[Standard voter model]\label{def:voter-model} %
  Let $\calG=(G_t\colon t\ge0)$ be a temporal graph, let $\kappa$ be a positive integer, and let $\sigma_0\colon V(G)\to\set{0,\dots,\kappa-1}$. The \emph{standard $\kappa$-opinion voter model} on~$\calG$ is a stochastic process $(\xi_t\colon t\ge0)$, where $\xi_t\colon V(\calG)\to\set{0,\dots,\kappa-1}$ is a random function that assigns an opinion to each vertex at time~$t$, and $\xi_0=\sigma_0$ holds deterministically for the initial state.

  At each time step~$t$, we sample $\xi_{t+1}(v)$ as follows: every vertex $v\in V(\calG)$ independently and synchronously chooses a neighbour $w$ in~$G_t$ uniformly at random, and adopts the opinion~$\xi_t(w)$ with probability~$1/2$; otherwise, $v$ keeps its current opinion~$\xi_t(v)$. This process continues until all vertices share the same opinion, which is called \emph{consensus}.
  The \emph{consensus time} is the first time~$t$ such that $\xi_t$ is a constant function.
\end{definition}

Note that $(\xi_t\colon t\ge0)$ is a time-inhomogeneous Markov chain on the state space of all functions $V(\calG)\to\set{0,\dots,\kappa-1}$, and that the absorbing states are precisely the constant functions.

Although our results extend to arbitrary~$\kappa$, the proof reduces to the two-opinion case\label{def:voter-model-two-opinion} [def:voter-model-two-opinion], so we mostly focus on $\kappa=2$.
In this setting, we write
\(A_t\)
for the set~\(\xi_t^{-1}(0)\) of vertices with opinion~$0$ at time~$t$, and we write
$\overline{A_t}\coloneqq V(\calG)\setminus A_t$ for the complement.
Consensus is reached when $A_t=\emptyset$ or $A_t=V(\calG)$ holds.

\subsection{Temporal conductance}\label{sec:conductance}

Before introducing a temporal analogue of conductance, we begin by recalling the standard definition of conductance for a static graph.
For a graph~$G$ with $m$~edges and a vertex set $S\subseteq V(G)$, we define the \emph{volume} of~$S$ and the \emph{conductance} of~$G$, via
\begin{align}\label{eq:volume} %
  \Vol(S) \coloneqq \sum_{v\in S} d(v)\,,
  &&
  \phi(G) \coloneqq \min_{\substack{S \subseteq V(G)\\0<\Vol(S)\leq m}} \frac{e_G(S,\overline{S})}{\Vol(S)}\,.
\end{align}
Here $\Vol(S)$ counts each internal edge of~$S$ twice and each edge between $S$ and~$\overline{S}$ once, so it measures the edge mass incident to~$S$; and $e_G(S,\ov{S})$ denotes the number of edges of~$G$ crossing the cut.
We remark that $m=\tfrac12\Vol(V(G))$ holds, so the condition $\Vol(S)\leq m$ is equivalent to $\Vol(S)\leq\tfrac12\Vol(V(G))$.
Among other things, the conductance determines how fast opinions can propagate across the cut in the voter model.

Now let $\calG=(G_t\colon t\ge0)$ be a temporal graph with fixed degrees.
The volume of a set $S\subseteq V(\calG)$ is time-invariant, but the edges crossing a cut are not.
For disjoint subsets $S,N\subseteq V(\calG)$, we write $e_t(S,N)$ for the number of edges of $G_t$ between $S$ and~$N$, and we abbreviate $e_t(v,N)\coloneqq e_t(\{v\},N)$ for a single vertex~$v$.
The \emph{conductance of the set~$S$ at time~$t$} is:
\begin{align}\label{eq:set-conductance-at-time} %
  \phi^t(S) \coloneqq \frac{e_t(S,\overline{S})}{\Vol(S)}
  \,.
\end{align}
The \emph{conductance of~$\calG$ on an integer interval $[a,b]\subseteq\N$} is defined by
\begin{align}\label{eq:conductance-on-interval} %
  \phi^{[a,b]}(\calG) \coloneqq \min_{\substack{S \subseteq
  V\\ 0<\Vol(S)\leq m}} \phi^{[a,b]}(S)
  &&\text{where}&&
  \phi^{[a,b]}(S) \coloneqq \max_{a\le t \le b} \phi^t(S)\,.
\end{align}
For $a=b=t$, we have $\phi^{[a,b]}(\calG)=\phi(G_t)$.
What $\phi^{[a,b]}(S)$ captures for $a<b$ is the \emph{best} conductance of $S$ across the interval $[a,b]$:
Indeed $\phi^{[a,b]}(S)\ge\phi$ holds if and only if $S$ has conductance $\phi^t(S)\ge\phi$ for at least one $t\in[a,b]$.
Moreover, $\phi^{[a,b]}(\calG)\ge\phi$ means that every set $S$ with $0<\Vol(S)\leq m$ has conductance $\phi^t(S)\ge\phi$ for at least one $t\in[a,b]$.
Note already here that $\phi^{[a,b]}(\calG)$ can be large even if $\phi(G_t)=0$ holds for all $t\in[a,b]$.

Then for any $\Delta \ge 1$, we define the \emph{$\Delta$-window conductance of~$\calG$} by
\begin{align}\label{eq:delta-window-conductance} %
  \phi^\Delta(\mathcal G) \coloneqq \min_{t\in\N}\phi^{[t,
  t+\Delta-1]}(\calG).
\end{align}
Now $\phi^\Delta(\mathcal G)\ge\phi$ means that, for every time interval of length $\Delta$, every set $S$ with $0<\Vol(S)\leq m$ has at least one ``good step'', that is, it has conductance at least $\phi$ in at least one graph $G_t$ in that interval.

Since $\phi^\Delta(\mathcal G)\le 1$ holds for all $\Delta$, we have $\phi^\Delta(\mathcal G)/\Delta\to 0$ as $\Delta\to\infty$, so the supremum of $\Delta\mapsto\phi^\Delta(\mathcal G)/\Delta$ is attained at some positive integer $\Delta^\ast$.
We can therefore define the \emph{conductance~$\Phi(\calG)$ of the temporal graph~$\calG$} via
\begin{align}\label{eq:temporal-conductance} %
  \Phi(\calG) \coloneqq \max_{\Delta \ge 1} \frac{\phi^\Delta(\calG)}{\Delta}\,.
\end{align}
If the temporal graph is static, that is, if we have $G_1=G_2=\cdots=G$, then $\Delta^\ast=1$ and thus $\Phi(\calG)=\phi(G)$, so the definitions coincide. The motivation for dividing by $\Delta$ in~\eqref{eq:temporal-conductance} is that the one ``good step'' in the length-$\Delta$ interval may be all we get before $S$ has been changed beyond recognition, and so we must amortise its effect over the entire interval length. At first this may seem like a very pessimistic assumption --- however, in the next section we provide a matching lower bound even in the case where we have lower bounds on the \emph{average} conductance within an interval (rather than the maximum), and so we argue that it is justified.

\paragraph{Example.} Recall our example from earlier --- suppose $\calG = (M_1,M_2,\dots)$, where each $M_i$ is a perfect matching drawn independently and uniformly at random. Then with high probability over the $M_i$'s, $\Phi(\calG) \ge \phi^3(\calG)/3 = \Omega(1)$, even though $\phi(G_t) = 0$ for all $t$. Effectively, for $\Phi^\Delta(\calG)$ to be high, we require that every set $S$ will be well-cut by \emph{some} graph in a length-$\Delta$ window, but this ``good step'' may differ within the window from one set to the next.

\subsection{Our Model and Results}

We prove bounds on the expected consensus time for the \emph{standard voter model} in a temporal undirected graph whose vertex degrees are fixed by using our notion of temporal conductance.
\begin{restatable}[Upper bound]{mtheorem}{UpperBoundf}\label{thm:upper-bound-temporal-conductance} [thm:upper-bound-temporal-conductance]
  There exists a constant $b'>0$ such that the following holds. Consider the standard voter model on an $n$-vertex temporal graph $\calG$ with fixed vertex degrees, $m$ edges, and minimum degree $d_{\min}$, with arbitrary initial state. Then with probability at least $1/2$, the consensus time is at most $b'm/(d_{\min}\Phi(\calG))$.
\end{restatable}

\Cref{thm:upper-bound-temporal-conductance} is our simple bound advertised in the introduction. Observe that applying it to the uniform matching example described in~\cref{sec:conductance} yields $O(n)$ absorption time, despite this example being disconnected at every time step.

\cref{thm:upper-bound-temporal-conductance} is a consequence of the following more general statement, which works with an arbitrary partition of time into windows of lengths $\Delta_0,\Delta_1,\dots$ and the conductance $\phi^{I_j}(\calG)$ attained on each window. One recovers \cref{thm:upper-bound-temporal-conductance} by choosing the windows that realise~$\Phi(\calG)$.
\begin{restatable}[Upper Bound]{theorem}{UpperBound}\label{thm:upper-bound-window-conductance} [thm:upper-bound-window-conductance]%
  There exists a constant $b>0$ such that the following holds. Consider the standard voter model on a temporal graph $\calG$ with fixed vertex degrees, $m$ edges, and minimum degree $d_{\min}$, with arbitrary initial state. Let $\Delta_0,\Delta_1,\dots \ge 1$ be integers, and let $\phi_0,\phi_1,\dots$ be real numbers in $[0,1]$. For all $j \ge 0$, let $I_j^- = \Delta_0 + \dots + \Delta_{j-1}$, let $I_j^+ = I_j^- + \Delta_j - 1$, and let $I_j = [I_j^-, I_j^+]$. Suppose that for all $j \ge 0$, $\phi^{I_j}(\calG) \ge \phi_j$. Let $J = \min\{j \colon \sum_{\ell=0}^j \phi_\ell \ge bm/d_{\min}\}$ (where we require $J<\infty$). Then with probability at least $1/2$, the consensus time is at most $\Delta_0 + \dots + \Delta_J$.
\end{restatable}

While the statement of \ref{thm:upper-bound-temporal-conductance} is shorter and the two are often equivalent (as in the uniform matching example above), we emphasise that \cref{thm:upper-bound-temporal-conductance} is stronger and more general.

These upper bounds are essentially tight: the dependence on~$\Phi(\calG)$ cannot be improved in general, as our matching lower bounds show.
\begin{restatable}[Lower bound]{mtheorem}{LowerBounds}\label{thm:lower-bound-3regular} [thm:lower-bound-3regular]%
  There exist arbitrarily large $n$-vertex 3-regular temporal graphs~$\calG$, and constant $c$ with the following property. From some initial state, with probability at least $1/2$, the standard voter model running on $\calG$ does not absorb within time $cn/\Phi(\calG)$.
\end{restatable}

The construction behind \cref{thm:lower-bound-3regular} in fact yields a stronger obstruction: it rules out fast consensus even against a weaker, window-averaged notion of conductance. In general, even if $I$ is a long interval, all that is needed to make $\phi^I(S)$ large is a single step $t$ with $\phi^t(S)$ large. It would be reasonable to expect \cref{thm:lower-bound-3regular} to fail if we required high \emph{average} conductance rather than high \emph{maximum} conductance in our definitions; thus in order to guarantee $\phi^I(\calG) \ge \phi$, we would require that for all $S \subseteq V(G)$ with $0 < \Vol(S) \le m$, $\tfrac{1}{|I|}\sum_{t \in I}\phi^t(S) \ge \phi$ rather than $\max_{t \in I}\phi^t(S) \ge \phi$. Indeed, we as authors originally expected that this would be the case, and that this redefinition would allow for a version of \cref{thm:upper-bound-window-conductance} with improved dependence on $\Delta_1,\dots,\Delta_J$; for example, defining $J = \min\{j \colon \sum_{\ell=0}^j \phi_\ell\Delta_\ell \ge bm/d_{\min}\}$ in the theorem statement. Such a result could have motivated a definition of temporal conductance with no dependence on interval lengths. However, \cref{thm:lower-bound-odd-regular-averaged} shows that no such result can hold in general, and that the dependence on interval lengths is in some sense intrinsic.
\begin{restatable}[Lower bound]{theorem}{LowerBoundf}\label{thm:lower-bound-odd-regular-averaged} [thm:lower-bound-odd-regular-averaged]%
  There exists $c>0$ such that for all odd integers $d$, there exist functions $\phi\colon\N\to\R$ and $\Delta\colon\N\to\N$ such that: as $n\to \infty$, $\phi(n) = \Theta(1/n)$ and $\Delta(n) = \Theta(n)$; and the following holds. There exist arbitrarily large $n$-vertex $d$-regular temporal graphs $\calG$ such that:
  \begin{enumerate}[(i)]
    \item For all time intervals $I$ of length $\Delta(n)$ and all sets $S \subseteq V(\calG)$ with $0 < \Vol(S) \le m$, we have $\tfrac{1}{\abs{I}}\sum_{t \in I} \phi^t(S) \ge \phi(n)$.
    \item From some initial state, with probability at least $1/2$, the standard voter model on $\calG$ does not absorb within time $cn\Delta(n)/(d\phi(n))$.
  \end{enumerate}
\end{restatable}

We focus on the case $\Delta(n)=\omega(1)$ in \cref{thm:lower-bound-odd-regular-averaged} in order to distinguish the criterion $\tfrac{1}{|I|}\sum_{t \in I}\phi^t(S) \ge \phi$ from the weaker criterion of $\max_{t \in I}\phi^t(S) \ge \phi$.

These results should be read against those of~\cite{DBLP:conf/icalp/BerenbrinkGKM16}. All their results are proved in the same adversarial model: the adversary is allowed to rewire the entire graph at each time step subject to fixed vertex degrees, and the results apply to the total conductance summed over all time steps. Naturally, their main upper bound of \cite[Theorem~1.1(i)]{DBLP:conf/icalp/BerenbrinkGKM16} transfers immediately to the non-adaptive setting while their main lower bound of \cite[Theorem~1.2]{DBLP:conf/icalp/BerenbrinkGKM16} breaks in the non-adaptive setting.

This important difference aside, \cite[Theorem~1.1(i)]{DBLP:conf/icalp/BerenbrinkGKM16} is precisely the restriction of \cref{thm:upper-bound-window-conductance} to length-1 intervals, while \cref{thm:upper-bound-temporal-conductance} can be read as a restriction of \cref{thm:upper-bound-window-conductance} to $\phi_0=\phi_1=\dots=\phi$ and $\Delta_0 = \Delta_1 = \dots = \Delta$, with optimised choices of $\phi$ and $\Delta$. Meanwhile, \cite[Theorem~1.2]{DBLP:conf/icalp/BerenbrinkGKM16} provides a graph for which the upper bound is tight with interval length $1$, while the interval length of \cref{thm:lower-bound-odd-regular-averaged} is $\Theta(n)$; otherwise, they both provide tight lower bounds in their respective settings. Throughout we retain the fixed-degree assumption of~\cite{DBLP:conf/icalp/BerenbrinkGKM16}.

\input{sections/1b_arxiv_proof_sketch}

%% file: sections/1b_arxiv_proof_sketch.tex
\subsection{Proof sketches}

We now sketch the proofs of our two main results.
\Cref{sec:upper-sketch} treats the upper bound (\cref{thm:upper-bound-temporal-conductance,thm:upper-bound-window-conductance}): a drift analysis of a Lyapunov potential that generalises the argument of Berenbrink, Giakkoupis, Kermarrec, and Mallmann-Trenn~\cite{DBLP:conf/icalp/BerenbrinkGKM16} by splitting time into \emph{stable} and \emph{unstable} intervals.
\Cref{sec:lower-sketch} treats the lower bound (\cref{thm:lower-bound-3regular,thm:lower-bound-odd-regular-averaged}), constructing temporal graphs on which the voter model is slow to reach consensus and explaining why the upper bound remains tight even when the conductance is averaged over each window rather than maximised.

\subsubsection{Upper bound}\label{sec:upper-sketch}

Let $\calG$ be an $n$-vertex $m$-edge temporal graph with fixed vertex degrees and minimum degree $d_{\min}$. For simplicity, fix $\phi$ and $\Delta$, divide time into intervals $I_0,I_1,\dots$ of length $\Delta$, and suppose that $\phi^{I_\ell}(\calG) \ge \phi$ for all $\ell\ge 0$. Then our goal is to prove that with probability at least $1/2$, the standard voter model reaches consensus within time $O(m\Delta/(d_{\min}\phi))$. (\cref{thm:upper-bound-window-conductance} allows for arbitrary interval lengths and conductance bounds, but this does not substantially affect the argument.) Again for simplicity, we take only $\kappa=2$ opinions; with some effort the argument of Berenbrink, Giakkoupis, Kermarrec, and Mallmann-Trenn~\cite{DBLP:conf/icalp/BerenbrinkGKM16} to bootstrap from two opinions to arbitrarily many can be adapted to our setting, so this is the heart of the proof.

We first review the $\Delta=1$ case treated in~\cite{DBLP:conf/icalp/BerenbrinkGKM16}, which is based on a drift analysis of a Lyapunov potential. (This is a standard method --- given a system $(X_t\colon t \ge 0)$ with state space $\Sigma$, filtration $(\calF_t\colon t \ge 0)$, and a potential function $f\colon\Sigma\to\R$, the \emph{drift} at time $t$ refers to $\E(f(X_{t+1})-f(X_t)\mid \calF_t)$. See e.g.~\cite{DBLP:series/ncs/Lengler20} for an overview.) Recall from \cref{def:voter-model} that $(A_t\colon t\ge0)$ is the evolution of the set of vertices with opinion~$0$ in the two-opinion voter model. In analysing the process, it is convenient to instead follow the \emph{minority set $S_t$ at time $t$} defined via
\begin{align*}
  S_t & \coloneqq
  \begin{cases}
    A_t            & \text{if $\Vol(A_t)\leq\Vol(\overline{A_t})$}; \\
    \overline{A_t} & \text{otherwise}.
  \end{cases}
\end{align*}
This trades the two absorbing states of $(A_t\colon t\ge0)$, $\emptyset$ and $V(G)$, for a single absorbing state at $\emptyset$. Note the minority set is chosen based on $\Vol(A_t)$, not $|A_t|$, for reasons we explain shortly.

The potential function used in~\cite{DBLP:conf/icalp/BerenbrinkGKM16} is given by $\psi(S_t) \coloneqq \sqrt{\Vol(S_t)}$. This is motivated by the fact that $(\Vol(A_t)\colon t \ge 0)$ is a martingale absorbing at $0$ and $\Vol(V(\calG))$. Since we define $S_t$ in terms of volume rather than size, it is not hard to show that $(\Vol(S_t)\colon t \ge 0)$ is a supermartingale absorbing at $0$, and this is the reason for the definition. The square root is then a simple concave function which turns a mostly-neutral drift into a negative one; this is quantified in~\cite[Lemma~2.1]{DBLP:conf/icalp/BerenbrinkGKM16}, which implies
\[
    \E[\psi(S_{t+1}) - \psi(s_t) \mid S_t=s_t] \le -\frac{d_{\min}e_t(s_t,\overline{s_t})}{32\psi(s_t)^3} = -\frac{d_{\min}\phi^t(s_t)}{32\psi(s_t)} \le -\frac{d_{\min}\phi}{32\sqrt{m}}.
\]
(Here the last inequality follows since $\Vol(s_t) \le \Vol(V(\calG))/2 = m$.) Since $\psi(S_0) \le \sqrt{m}$ and at each step we have expected drift $-\Omega(d_{\min}\phi/\sqrt{m})$, we should therefore expect absorption in $O(\sqrt{m}/(d_{\min}\phi/\sqrt{m})) = \Theta(m/(d_{\min}\phi))$ time, and it is not hard to make this bound rigorous using standard martingale methods.

The difficulty of extending this proof to our setting lies in a subtle but important point. If we could guarantee that for all $i \ge 0$, there existed $t \in I_i$ such that $\phi^t(S_t) \ge \phi$, then we would be able to use essentially the same argument as in~\cite{DBLP:conf/icalp/BerenbrinkGKM16}: we would accrue $\phi/\sqrt{m}$ negative drift from our ``one good step'' of $t \in I_i$, and at worst accrue no negative drift from the remaining steps in $I_i$. However, this is not the case. Writing $I_i^-$ for the start of interval $I_i$, we can only guarantee that there exists $t \in I_i$ such that $\phi^t(S_{I_i^-}) \ge \phi$; $t$ and $I_i^-$ may be far apart and $S_t$ may be very different from $S_{I_i^-}$, so this is not sufficient. This is not merely a technical difficulty --- we will see in \cref{sec:lower-sketch} that this situation is the reason why our upper bound is tight even given a lower bound on the average conductances $\frac{1}{|I_i|}\sum_{t \in I_i}\phi^t(S)$ rather than simply the maximum conductances $\phi^{I_i}(S)$.

We effectively proceed by a case analysis on each interval $I_i$, conditioned on the state of the process at $I_i^-$. Let $t' \in I_i$ be the ``one good step'' we are guaranteed with $\phi^{t'}(S_{I_i^-}) \ge \phi$. Very informally, our cases are:
\begin{enumerate}[(i)]
    \item Neither $\Vol(S_t)$ nor $e_t(S_t,\overline{S_t})$ is likely to change significantly over $t \in I_i$.
    \item $\Vol(S_t)$ is not likely to change significantly over $t \in I_i$, but $e_t(S_t,\overline{S_t})$ is likely to change significantly over $t \in I_i$.
    \item $\Vol(S_t)$ is likely to change significantly over $t \in I_i$.
\end{enumerate}
In case (i), we have $\phi^{t'}(S_{t'}) \approx \phi^{t'}(S_{I_i^-})$ and so we may apply the same analysis as in~\cite{DBLP:conf/icalp/BerenbrinkGKM16}. In case (ii), informally, we will prove that when volume is held roughly constant, $e(S_t,\overline{S_t})$ can only be likely to change significantly over $t \in I_i$ if $\sum_{t \in I_i} e(S_t,\overline{S_t})$ is likely to be large; since volume is held roughly constant, this implies that $\sum_{t \in I_i}\phi^t(S_t)$ is likely to be large and we again have significant negative drift. In case (iii), while we may or may not have significant negative drift in $\psi(S_t)$, it is at least true that $\Vol(S_t)$ is likely to change significantly; we might expect that absorption should be likely after relatively few such intervals have passed.

In order to make these ideas rigorous, in cases (i) and (ii) we will need deterministic bounds on $\Vol(S_t)$. To obtain these, rather than working with the intervals $I_0,I_1,\dots$ directly, we define stopping times $0=T_0,T_1,\dots$ in \cref{def:stopping-times} and work with the intervals $[T_i,T_{i+1})$. By default each stopping time $T_i$ occurs at the start of interval $I_i$, but it can occur early at time $t>T_{i-1}$ if $\Vol(S_t) \notin [\tfrac12\Vol(S_{T_{i-1}}), \tfrac32\Vol(S_{T_{i-1}})]$. Thus $\Vol(S_t)$ is guaranteed to stay near-constant within each interval $[T_i,T_{i+1})$.

Given this formulation, the distinction between cases (i) and (ii) versus case (iii) for a given interval $[T_i,T_{i+1})$ is whether or not $T_{i+1}$ is likely to occur early (i.e.\ before the start of interval $I_{i+1}$). We set this out in \cref{def:stable-unstable-interval}, where we define a \emph{stable} interval $[T_i,T_{i+1})$ as one with
\[
    \E[|\Vol(S_{T_{i+1}}) - \Vol(S_{T_i})|\mid \calH_{T_i}] < \Vol(S_{T_i})/8,
\]
where $\calH_t$ generates the filtration of the voter model and is formally defined in \cref{sec:martingale-properties}. An \emph{unstable} interval is one for which this does not hold. By Markov's inequality we see that $T_{i+1}$ is unlikely to occur early in a stable interval, as will be required for cases (i) and (ii); for case (iii) in an unstable interval, while the lower bound on expected change in volume is technically weaker than a probabilistic bound, it will turn out to be more natural to work with.

We perform a drift analysis on stable intervals with respect to $\psi$ in \cref{sec:case-1}. Our key result in both cases is \cref{cor:potdec-regular}, which exploits the fact that $\Vol(S_t)$ cannot change significantly over $I_i$ to bound the backward drift of $\psi$ in terms of
\[
    \mu \coloneqq \E\Big[\sum_{t=T_i}^{T_{i+1}-1} e_t(S_t,\overline{S_t}) \mid \calH_{T_i}\Big].
\]
In order for \cref{cor:potdec-regular} to give meaningful negative drift in $\psi$, we will need to bound $\mu$ below. Recall that by our hypothesis, we are guaranteed ``one good step'' $t' \in I_i$ for which $\phi^{t'}(S_{T_i}) \ge \phi$. We bound $\mu$ below in \cref{lem:prob-good-event} via a case distinction on
\[
    \mu'\coloneqq \E\Big[\sum_{t=T_i}^{\min\{T_{i+1},t'\}-1} e_t(S_t,\overline{S_t}) \mid \calH_{T_i}\Big].
\]
If $\mu'$ is small, then via \cref{lem:stepwise-edges-bound} it can be shown that we are in case (i) --- that is, in expectation $e_{t'}(S_{t'},\overline{S_{t'}}) \approx e_{t'}(S_t,\overline{S_t})$ and hence $\phi^{t'}(S_{t'}) \gtrapprox \phi$. (This should match our intuition --- if $e_t(S_t,\overline{S_t})$ is small for some $t$, then we should expect $S_{t+1}$ to be close to $S_t$ and hence $e_{t'}(S_{t+1},\overline{S_{t+1}}) \approx e_{t'}(S_t,\overline{S_t})$.) We can use this to bound $\mu$ below at the $t'$ term. If instead $\mu'$ is large, then we are in case (ii), and we can bound $\mu$ below directly by $\mu'$. Combining \cref{cor:potdec-regular} and \cref{lem:prob-good-event} yields \cref{lem:psi-down-drift}, which implies that $(\psi(S_{T_i})\colon i\geq0)$ is a supermartingale with drift $-\Omega(d_{\min}\phi/\sqrt{m})$ whenever $[T_i,T_{i+1})$ is stable.

Having finished with cases (i) and (ii), we next analyse case (iii) (when the interval $[T_i,T_{i+1})$ is unstable) in \cref{sec:case-2}. Here we forego an analysis of $\psi$ in favour of a different potential function. Intuitively, if every interval $[T_i,T_{i+1})$ were unstable, then on average we should expect each $\Vol(S_{T_{i+1}})$ to differ from $\Vol(S_{T_i})$ by at least a constant factor, giving rise to something like a simple random walk on $\{0,\dots,\log (2m)\}$. We are able to make this idea rigorous by analysing the drift of the potential function $\chi(S_t) \coloneqq \log(1+\Vol(S_t))$. We do so first in a generic setting (\cref{lem:chi-down-drift-generic}) before proving in \cref{cor:chi-down-drift-voter} that $(\chi(S_{T_i})\colon i\geq0)$ is a supermartingale which has drift $-\Omega(1)$ whenever $[T_i,T_{i+1})$ is an unstable interval.

It now remains to combine the stable case with the unstable case, which we do in \cref{sec:combining-cases}. Thankfully, we proved that both $(\psi(S_{T_i})\colon i\geq0)$ and $(\chi(S_{T_i})\colon i \geq0)$ are supermartingales regardless of which intervals are stable or unstable. As such we can simply normalise and add our two potentials to obtain a final supermartingale 
\[
\Psi_i \coloneqq \frac{\sqrt{m}}{d_{\min}}\psi(S_{T_i}) + \chi(S_{T_i}).
\]
When $[T_i, T_{i+1})$ is stable, $\Psi$ has $-\Omega(\phi)$ drift from $\psi$, and when $[T_i, T_{i+1})$ is unstable it has $-\Omega(1)$ drift from $\chi$; thus in both cases $\Psi$ has $-\Omega(\phi)$ drift. Since $\Psi_i = O(m/d_{\min})$ for all $i$, applying the optional stopping theorem in the usual way then yields our two-opinion result of \cref{thm:voter-absorb-two-opinion}. We can then bootstrap to our final $\kappa$-opinion result in a similar fashion to~\cite{DBLP:conf/icalp/BerenbrinkGKM16}. (Strictly speaking this requires a bound on absorption time in terms of $S_0$ rather than $m$, but this does not require significant changes to the proof sketch.)

\subsubsection{Lower bound}\label{sec:lower-sketch}

Let $n \ge 4$ be an even integer, and let $T \ge \log^2 n$ be an integer. We sketch the construction of a $1$-regular $n$-vertex temporal graph $\calG$ with $\phi(\calG) = \Omega(1/(nT))$, average (static) conductance $\Omega(1/n)$ on all intervals of length $3T$, and upon which the standard two-opinion voter model may take $\Omega(Tn^2)$ time to absorb. This suffices to prove \cref{thm:lower-bound-odd-regular-averaged} for $d=1$, as well as \cref{thm:lower-bound-3regular} in general. We will then explain how to adapt the construction to other vertex degrees.

Let $V(\calG) = [n]$. For all $i \ge 0$, let $I_i = [iT, (i+1)T-1]$. For all $t \in I_1 \cup I_3 \cup \dots$, let the edge set of $\calG_t$ be $\{\{1,2\}, \{3,4\}, \dots, \{n-1,n\}\}$. For all $t \in I_0 \cup I_2 \cup \dots$, let the edge set of $\calG_t$ be $\{\{2,3\},\{4,5\}\dots,\{n-2,n-1\},\{n,1\}\}$. Thus $\calG$ alternates with period $T$ between the two perfect matchings of an $n$-vertex cycle; in particular, $\calG$ is $1$-regular as required. 

We next bound the conductance of $\calG$ below. Fix $S \subseteq V(\calG)$ with $0 < \Vol(S) \le \tfrac12\Vol(V(\calG)) = n/2$. Observe that since $S \notin \{\emptyset, V(\calG)\}$, $S$ must send out at least one edge throughout either odd or even intervals $I_i$ (or both); suppose without loss of generality that $e_t(S,\overline{S}) \ge 1$ for all $t \in I_i$ with $i$ odd. Then any interval $I$ of length $3T$ must contain at least one such $I_i$, and so
\[
    \frac{1}{3T}\sum_{t \in I}\phi^t(S) \ge \frac{1}{3T}\sum_{t \in I_i}\phi^t(S) \ge \frac{1}{3T}\sum_{t \in I_i}\frac{1}{\Vol(S)} \ge \frac{1}{3T}\cdot \frac{|I_i|}{n/2} = \frac{2}{3n}.
\]
Thus the required lower bound on average static conductance holds. Moreover, by a similar argument, every length-$(T+1)$ interval contains at least one step $t$ with $\phi^t(S) \ge 2/n$, so $\Phi(\calG) = \Omega(1/(nT))$. While not necessary for the proof, it is not hard to show that $\Phi(\calG) = O(1/(nT))$ also by considering sets $S$ of the form $\{1,\dots,k\}$ with $k = \Theta(n)$ even.

Finally, we prove that the standard two-opinion voter model may take $\Omega(Tn^2)$ time to absorb. Suppose the initial opinion-$0$ set is of the form $\{2,\dots,k\}$ with $k=\Theta(n)$ even. Then throughout $I_0$, the only opinions which can change are those of $k$ and $k+1$. Recall that the standard voter model is synchronous and lazy, so at each time step one vertex takes the other's opinion with probability $1/2$; otherwise, either they exchange opinions (with probability $1/4$) or do nothing (with probability $1/4$). Since $|I_0|=T \ge \log^2 n$, with probability $1-n^{-\omega(1)}$, at the end of $I_0$, $k$ and $k+1$ will share the same opinion and the new opinion-$0$ state will be either $\{2,\dots,k-1\}$ or $\{2,\dots,k+1\}$. The process continues to evolve in this manner: with high probability, over the course of each interval $I_i$, the opinion-$0$ set changes from one contiguous interval to another, with each endpoint shifting by at most $1$ in an unbiased fashion. It is not hard to prove (e.g.\ using Azuma's inequality) that with probability at least $1/2$, such a system does not absorb until $\Omega(n^2)$ intervals $I_i$ have passed, i.e.\ $\Omega(Tn^2)$ total time.

For the more general construction (for arbitrary odd vertex degrees), we blow up each vertex $i$ into a set $V_i$ of size $(d-1)/2$ and replace each matching edge $\{i,i+1\}$ (modulo $n$) with a clique on $V_i \cup V_{i+1}$, obtaining an $nd$-vertex graph. This construction is defined as $\calG^{T,k,z}$ in \cref{def:lower-bound-construction}. The bounds on conductance and absorption time follow by analogous arguments to the $d=1$ case in Lemmas~\ref{lem:clique-lower-bound-conductance} and~\ref{lem:lower-bound-absorption} respectively, and the fundamental behaviour of the construction is similar (with the opinion-$0$ state being of the form $V_i \cup V_{i+1} \cup \dots \cup V_j$ with high probability at the start of every interval). 

%% file: sections/2_preliminaries.tex
\section{Preliminaries}
In this section we introduce the notation, basic definitions and well-known probability results relevant to prove our main theorems. Throughout the paper, upper-case letters such as $A_t$ denote random variables taking values in subsets of $V(\calG)$, whereas lower-case letters such as $a_t$ denote fixed subsets of $V(\calG)$. We also use the notation $a\wedge b\coloneqq\min\{a,b\}$ and $a\vee b\coloneqq\max\{a,b\}$. 
We use $\log$ to denote the natural logarithm.

\subsection{General Lemmas}
Below we state two classical results that will be used repeatedly throughout the paper.

\begin{lemma}[Optional stopping theorem]\label{lem:optional-stopping-time}\leanok%
    Let $(X_t\colon t\geq0)$ be a supermartingale with some filtration $(\genericfiltration{t}\colon t\geq0)$. Let $\tau$ be a stopping time for $(X_t\colon t\geq0)$, suppose that $\tau$ is almost surely finite, and let $t \ge 0$. Then $\E(X_\tau \mid \calF_t) \le X_t$.
\end{lemma}

\begin{lemma}[Azuma's inequality for
  martingales]\label{lem:azuma-inequality}\leanok%
  Let $(Y_t\colon t\geq0)$ be a martingale w.r.t the filtration
  $(\genericfiltration{t}\colon t\geq0)$.
  Let $n\in\N$ and suppose there exist constants $c_1,\dots,c_{n}\ge
  0$ such that
  \(
    \abs{Y_{i}-Y_{i-1}}\le c_i
  \)
  holds almost surely for every $i\in\{1,\dots,n\}$.
  Then, for every $T\ge 0$, we have:
  \[
    \pr[\abs{Y_n-Y_0}\ge T]\le 2\exp\left(-\frac{T^2}{2\sum_{i=1}^{n} c_i^2}\right).
  \]
\end{lemma}

\subsection{Martingale properties of the volume process}\label{sec:martingale-properties}

We now show that $(\Vol(A_t)\colon t\geq0)$ is a \emph{martingale} and $(\Vol(S_t)\colon t\geq0)$ is a \emph{supermartingale} with respect to the standard filtration of the voter model on $\calG$, which is defined as follows.
 \begin{definition}[Filtration]\label{def:filtration}\leanok%
    We define \emph{the filtration of the standard voter model} as $(\sigma(\calH_t) \colon t \geq0)$ where $\calH_t=(A_0,A_1,\dots,A_t)$ is the random variable that holds the entire \emph{history} of the process so far and $\sigma(\calH_t)$ is the $\sigma$-algebra generated by~$\calH_t$.
\end{definition}
A sequence $(X_t\colon t\in\N)$ of random variables is a
\emph{martingale under a filtration $(\genericfiltration{t}\colon t\in\N)$} if for all~$t\in\N$, we have $\E[X_t]< \infty$ and $\E[X_{t+1}\conditioned
\genericfiltration{t}]=X_t$.
The sequence is a \emph{sub-martingale} if instead
$\E[X_{t+1}\mid\genericfiltration{t} ]\geq X_t$ holds,
and it is a \emph{super-martingale} if
instead $\E[X_{t+1}\mid\genericfiltration{t}]\leq X_t$ holds.

In the standard voter model with two opinions competing on a temporal graph whose vertex-degrees are fixed, it is equally likely that a vertex with
opinion 1 switches to opinion 0 as it is that a vertex with opinion 0
switches to opinion 1.
Thus, the evolution of the volume of each opinion is a martingale.
\begin{lemma}\label{lem:volume_A_martingale}\leanok%
  Let $\calG$ be a temporal graph whose vertex-degrees are fixed.
  Let $(A_t\colon t\geq 0)$ be the evolution of the opinion-$0$ set in
  the standard voter model on~$\calG$. Then, $(\Vol(A_t) \colon t \geq0)$ and $(\Vol(\overline{A_t}) \colon t \geq0)$ are martingales under the filtration $\masterfiltration{t}$.
\end{lemma}
\begin{proof}
  We only show that $(\Vol(A_t)\colon t\geq0)$ is a martingale; the proof for $(\Vol(\overline{A_t})\colon t\geq0)$ is then immediate.
  We have to show ${\E[\Vol(A_{t+1})-\Vol(A_t)\mid \filtration{t}]=0}$.
  Note that $A_t$ and thus $\Vol(A_t)$ is fully determined by the history~$\filtration{t}$, but the randomness used to generate $A_{t+1}$ from~$A_t$ (that is, one synchronous update step of the voter process) is independent from~$\filtration{t}$.
  From the definition of the standard voter model, we thus have 
  \begin{align*}
        &\E\bigl[\,\Vol(A_{t+1})-\Vol(A_t)\bigm|\filtration{t}\,\bigr]\\
        &\qquad\qquad= \E\Big[\sum_{v \in A_{t+1}\setminus A_t}d(v) - \sum_{v \in A_t\setminus A_{t+1}} d(v) \,\Big|\,\calH_t\Big]\\
        &\qquad\qquad=\sum_{v\in\ov{A_t}}\pr(v\in A_{t+1}\mid \calH_t)d(v)
        -\sum_{v\in A_t} \pr(v\in \ov{A_{t+1}}\mid \calH_t)d(v)
        \\
        &\qquad\qquad=\sum_{v\in\ov{A_t}}\frac{e_t(v,A_t)}{2d(v)}d(v)-\sum_{v\in A_t} \frac{e_t(v,\ov{A_t})}{2d(v)}d(v)
        \\
        &\qquad\qquad=
        \tfrac{1}{2}\cdot \big(e_t(\ov{A_t},A_t)-e_t(A_t,\ov{A_t})\big)=0.
  \end{align*}
   The ultimate equality holds because the graph is undirected, which concludes the proof.
 \end{proof}

We now state the standard fact that the minimum of two martingales is a supermartingale.
\begin{lemma}\label{lem:from-martingale-to-supermartingale}\leanok%
  Let $(M_t\colon t\geq0)$ and $(N_t\colon t\geq0)$ be two discrete-time
  martingales with respect to the same filtration
  $(\genericfiltration{t}\colon t\geq0)$. Then $(X_t\colon t\geq0)$ defined
  via $X_t\coloneqq\min\{M_t,N_t\}$ is a supermartingale with respect to $(\genericfiltration{t} \colon t \geq0)$.
\end{lemma}

\begin{proof}
  Since $M_t$ and $N_t$ are $\genericfiltration{t}$-measurable, $X_t$ is
  $\genericfiltration{t}$-measurable. Moreover, for any $t\geq0$,
  $\E[\abs{M_t}]<\infty$ and $\E[\abs{N_t}]<\infty$  then $\E[X_t]<\infty$.
Since $\E[M_{t+1}\wedge N_{t+1} \mid \calF_t] \le \E[M_{t+1}\mid \calF_t]$ and likewise for $N_{t+1}$, we have
  \begin{equation*}
    \E\bigl[M_{t+1}\wedge N_{t+1}\bigm| \genericfiltration{t}\bigr]\leq
    \E[M_{t+1}\conditioned\genericfiltration{t}]\wedge\E[N_{t+1}\conditioned\genericfiltration{t}].
  \end{equation*}
   Since $M_t$ and $N_t$ are martingales, it follows that
   \begin{equation*}
       \E[M_{t+1}\wedge N_{t+1}\conditioned \genericfiltration{t}]\leq M_t\wedge N_t.
   \end{equation*}
   The result follows.
\end{proof}

 We recall the definition of the minority set from \cref{sec:upper-sketch}.

\begin{definition}\label{def:minority-set}
    Let $\calG$ be a temporal graph whose vertex-degrees are fixed.
  Let $(A_t\colon t\geq 0)$ be the evolution of the opinion-$0$ set in
  the standard voter model on~$\calG$. The \emph{minority set $S_t$ at time $t$} is given by
\begin{align*}
  S_t & \coloneqq
  \begin{cases}
    A_t            & \text{if $\Vol(A_t)\leq\Vol(\overline{A_t})$}; \\
    \overline{A_t} & \text{otherwise}.
  \end{cases}
\end{align*}
\end{definition}

Using the previous two lemmas, we show that the evolution of the volume of the minority set is a supermartingale.

\begin{corollary}\label{cor:vol-minority-supermartingale}\leanok%
  Let $\calG$ be a temporal graph whose vertex-degrees are fixed. Let
  $(S_t\colon t\geq0)$ be the evolution of the minority set at time
  $t$ in the standard voter model on $\calG$ with $\kappa=2$. Then,
  $(\Vol(S_t)\colon t\geq0)$ is a supermartingale with respect to $\masterfiltration{t}$.
\end{corollary}
\begin{proof}
  Let $A_t\subseteq V(\calG)$ be the set of vertices with opinion $0$
  at time $t$ and $\ov{A_t}$ the set of vertices with opinion $1$ at time $t$.
  By \cref{lem:volume_A_martingale}, we have that both
  $\Vol(A_t)$ and $\Vol(\ov{A_t})$ are martingales with respect to
  $\masterfiltration{t}$. Observe that
  $\Vol(S_t)=\min\{\Vol(A_t),\Vol(\ov{A_t})\}$, then by
  \cref{lem:from-martingale-to-supermartingale},  $\Vol(S_t)$ is a
  supermartingale with respect to $\masterfiltration{t}$.
\end{proof}

\subsection{Edges Crossing the Cut}

The following lemma focus on the change in the number of
edges between a given set and its complement, in particular, it
establishes a bound for the stepwise variation.

\begin{lemma}\label{lem:stepwise-edges-bound}\leanok%
  Suppose that $\calG$ is a temporal graph whose vertex-degrees are
  fixed.
  Let $(S_t \colon t \geq 0)$ be the
  evolution of the minority set in the standard voter model on $\calG$ with $\kappa=2$. For any $t,j\geq 0$, we have
  \begin{equation*}
    \E\Big[\abs*{e_t(S_{j+1},\ov{S_{j+1}}) - e_t(S_{j},\ov{S_{j}}) }
    \;\Big\vert\;  \filtration{j}\Big] \le e_{j}(S_{j},\ov{S_{j}}).
  \end{equation*}
\end{lemma}

\begin{proof}
  First observe that for all $x,y \ge 0$, $e_x(S_y,\ov{S_y}) = e_x(A_y,\ov{A_y})$; thus it suffices to prove
  \[
    \E\Big[\abs*{e_t(A_{j+1},\ov{A_{j+1}}) - e_t(A_{j},\ov{A_{j}}) }
    \;\Big\vert\;  \filtration{j}\Big] \le e_{j}(A_{j},\ov{A_{j}}).
  \]
  Let $ A_{j}\ominus A_{j+1}$ be the symmetric difference of $A_{j}$
  and $A_{j+1}$. Every single vertex $v$ that swaps from
  $A_{j}$ to $\ov{A_{j}}$ or vice versa removes at most $d(v)$ edges from
  the cut. This gives a bound of
  \[
    e_t(A_{j+1},\ov{A_{j+1}}) \ge e_t(A_{j},\ov{A_{j}}) - \sum_{v \in
    A_{j}\ominus A_{j+1}} d(v).
  \]
Likewise, every single vertex $v$ that swaps
  from $A_{j}$ to $\ov{A_{j}}$ or vice versa adds at most $d(v)$ edges
  to the cut, giving
  \[
    e_t(A_{j+1},\ov{A_{j+1}}) \le e_t(A_{j},\ov{A_{j}}) + \sum_{v \in
    A_{j}\ominus A_{j+1}} d(v).
  \]
  Therefore, the absolute change is bounded as follows,
  \begin{equation*}
    \abs{e_t(A_{j+1},\ov{A_{j+1}}) - e_t(A_{j},\ov{A_{j}})} \leq \sum_{v
    \in A_{j}\ominus A_{j+1}} d(v).
  \end{equation*}
  Hence
  \[
    \E\Big[\abs[\big]{e_t(A_{j+1},\ov{A_{j+1}}) - e_t(A_{j},\ov{A_{j}})}
    \Bigm|  \filtration{j}\Big] \le \E\Big[\sum_{v \in A_{j}\ominus
    A_{j+1}} d(v)\Bigm|  \filtration{j}\Big].
  \]
  Moreover,
  \begin{align*}
    \E\Big[\sum_{v \in A_{j}\ominus A_{j+1}} &d(v) \Bigm|
    \filtration{j}\Big] = \sum_{v \in V}\pr(v \in A_{j}\ominus A_{j+1} \mid \filtration{j})\cdot d(v) \\
    &= \sum_{v \in A_{j}}\pr(v \in \ov{A_{j+1}}\mid
    \filtration{j})\cdot d(v) + \sum_{v \in \ov{A_{j}}}\pr(v \in A_{j+1} \mid \filtration{j})\cdot d(v)\\
    &= \sum_{v \in A_{j}}
    \frac{e_{j}(v,\overline{A_{j}})}{2d(v)}\cdot d(v) + \sum_{v
    \in \ov{A_{j}}}\frac{e_{j}(v,A_{j})}{2d(v)}\cdot d(v)\\
    &= \frac{1}{2}e_{j}(A_{j},\ov{A_{j}}) +
    \frac{1}{2}e_{j}(A_{j},\ov{A_{j}})\\
    &= e_{j}(A_{j},\ov{A_{j}})\,.
  \end{align*}
  This finishes the proof.
\end{proof}

%% file: sections/3_main_result.tex
\section{Proof of the Main Theorem}\label{sec:main-proof}%
In this section, our goal is to prove \cref{thm:upper-bound-window-conductance}. The heart of the proof is the two-opinion case; given this, we will be able to follow the argument of Berenbrink, Giakkoupis, Kermarrec, and Mallmann-Trenn~\cite{DBLP:conf/icalp/BerenbrinkGKM16} to extend to the general case. Our result is the following, which is essentially the two-opinion case.

\begin{restatable}{theorem}{UpperBoundTwoOpinion}\label{thm:voter-absorb-two-opinion}\leanok[thm:voter-absorb-two-opinion]%
    There exists a constant $b>0$ such that the following holds. Consider the standard two-opinion voter model on a temporal graph $\calG$ with fixed vertex degrees and minimum degree $d_{\min}$, with arbitrary initial 
    minority set $s_0$. Let $\Delta_0,\Delta_1,\dots \ge 1$ be integers, and let $\phi_0,\phi_1,\dots$ be real numbers in $[0,1]$. For all $j \ge 0$, let $I_j^- = \Delta_0 + \dots + \Delta_{j-1}$, let $I_j^+ = I_j^- + \Delta_j - 1$, and let $I_j = [I_j^-, I_j^+]$. Suppose that for all $j \ge 0$, $\phi^{I_j}(\calG) \ge \phi_j$. Let
    \[
    J = \min\Big\{j \colon \sum_{\ell=0}^j \phi_\ell \ge b\Big(\frac{\Vol(s_0)}{d_{\min}} + \log(1+\Vol(s_0))\Big)\Big\},
    \]
    where we require $J < \infty$. Then with probability at least $1/2$, the consensus time is at most $\Delta_0+\dots+\Delta_J$.
\end{restatable}

Note that when $\Vol(s_0) = \Theta(m)$ in \cref{thm:voter-absorb-two-opinion}, the $\log(1+\Vol(s_0))$ term is dominated by the $m/d_{\min} = \Omega(n)$ term.

We will prove \cref{thm:voter-absorb-two-opinion} in \cref{sec:combining-cases}, then prove \cref{thm:upper-bound-window-conductance} from \cref{thm:voter-absorb-two-opinion} in \cref{sec:multi-opinion}. To this end, until the end of \cref{sec:combining-cases}, let the following be defined as in the statement of \cref{thm:voter-absorb-two-opinion}: $\calG$, $d_{\min}$, $J$, and (for all $j$) $\Delta_j$, $\phi_j$, $I_j^+$, $I_j^-$ and $I_j$.
 Let $S_0\subseteq V$ be a fixed initial minority
set, and let $(S_t\colon t\ge 0)$ be the minority set of a standard
voter model running on $\calG$ from $S_0$ with $\kappa=2$ opinions.

We will bound the likely absorption time of $(S_t)$ by considering an embedded voter process. One might expect this voter process to consider times $I_1^-,I_2^-,\dots$, but it will be convenient for us to ``stop an interval early'' if the volume of $(S_t)$ deviates significantly from its initial value. As such we define the following sequence of stopping times.

\begin{definition}[Stopping times]\label{def:stopping-times}\leanok%
  Let $T_0 = 0$. Given $T_0,\dots,T_j$ for some $j \ge 0$, we define $T_{j+1}$ as follows. Let 
  \[T_{j_{\min}}\coloneqq\min\big\{t \geq T_j \colon
          \Vol(S_t)\notin \Big[\tfrac{1}{2}\Vol(S_{T_j}),\;
      \tfrac{3}{2}\Vol(S_{T_j})\Big] \textnormal{ or }\; S_t = \emptyset\big\}\,.\]
  Then
  \[T_{j+1}\coloneqq\min\set[\Big]{I_{j}^+ + 1,\; T_{j_{\min}}}\,.\]      
\end{definition}

Observe that $T_{j_{\min}}$ and $T_j$ are indeed stopping
times for all $j$.\label{stmt:are-stopping-times}\leanok%

\begin{definition}[Embedded voter process]\label{def:embedded-voter-process}\leanok%
  The \emph{embedded voter process} is the process $((T_j, S_{T_j})\colon
  j\geq0)$ where the minority set is sampled only at the stopping times
  defined in \cref{def:stopping-times}. Moreover, the \emph{embedded
  volume process} is defined by $(\Vol(S_{T_j})\colon j\geq0)$.
\end{definition}
Observe that the embedded voter process is a Markov chain and that the embedded volume process is a supermartingale, from the corresponding properties of the original voter process (together with the optional stopping theorem).
 (If $T_j$ were not included in the $j$'th state, however, then the embedded voter process would not be a Markov chain --- indeed, the value of each graph $G_t$ on which the voter process runs in $[T_j,T_{j+1}-1]$ depends on $t$ itself, not only on $t-T_j$).

Recall from \cref{thm:upper-bound-window-conductance} that our goal is to prove likely absorption in $(S_t)$ within time $\Delta_0 + \dots + \Delta_J = I_{J}^+ + 1$; since $T_J \le I_J^+ + 1$, it will suffice to prove likely absorption within time $T_J$, and we will focus our efforts on the embedded chain from this point forward. We will analyse the embedded chain by splitting the intervals $[T_j, T_{j+1}-1]$ into two kinds depending on whether or not $\Vol(S_t)$ is likely to fluctuate significantly over $t \in [T_j, T_{j+1}-1]$.

\begin{definition}[Stable and unstable
  intervals]\label{def:stable-unstable-interval}\leanok%
  Let $j\in\N$, and let $\filtrationValue{T_j}$ be a possible value of
  $\filtration{T_j}$ such that $S_{T_j} \ne \emptyset$ whenever $\filtrationValue{T_j} =
  \filtration{T_j}$. If $\E\bigl[\abs{\Vol(S_{T_{j+1}}) - \Vol(S_{T_j})} \bigm|
  \filtration{T_j} = H_{T_j}\bigr] < \Vol(S_{T_j})/8$, we say that both $\filtrationValue{T_j}$ and
  the interval $[T_j, T_{j+1})$ are \emph{stable}. Otherwise, we say
  they are both \emph{unstable}.
\end{definition}

We will analyse stable intervals in \cref{sec:case-1}, then analyse unstable intervals in \cref{sec:case-2}, then combine the two cases in \cref{sec:combining-cases}.

\input{sections/3.1_case_1}
\input{sections/3.2_case2}
\input{sections/3.3_interplay}
\input{sections/3.4_multi_opinion}

%% file: sections/3.1_case_1.tex
\subsection{Stable intervals}\label{sec:case-1}%

In this section we prove backwards bias in $\psi$ conditioned on
$\filtration{T_j} = \filtrationValue{T_j}$ where $\filtrationValue{T_j}$ is stable, so that
$\E\bigl[\,\abs{\Vol(S_{T_{j+1}}) - \Vol(S_{T_j})} \bigm|  \filtration{T_j} = H_{T_j}\,\bigr] <
\frac{1}{8}\Vol(S_{T_j})$. We do this as
\cref{lem:psi-down-drift}. We first import the following
result from~\cite{DBLP:conf/icalp/BerenbrinkGKM16} which guarantees
backward bias at each step of the standard (non-embedded) voter model.

\begin{lemma}\label{lem:potdec-regular}\leanok%
  Let $t \ge 0$, let $\filtrationValue{t}$ be a possible value of $\filtration{t}$, and let
  $s_t$ be the value of $S_t$ determined by $\filtration{t}$. Suppose that
  $s_t \ne \emptyset$. Then 
  \[
    \E[\psi(S_{t+1}) \, \lvert \,  \filtration{t}=\filtrationValue{t}] \le \psi(s_t) -
    \frac{\dmin}{32}\cdot\frac{e_t(s_t,\overline{s_t})}{\psi(s_t)^3}.
  \]
\end{lemma}
\begin{proof}
  We first import some
  notation from~\cite{DBLP:conf/icalp/BerenbrinkGKM16}. For a vertex $u \in s_t$, let $\lambda_{u,t} \coloneqq
  e_t(u,\overline{s_t})$, and conversely for a vertex $u \in
  \overline{s_t}$ let $\lambda_{u,t} \coloneqq e_t(u,s_t)$. Then
  \cite[Lemma~2.1]{DBLP:conf/icalp/BerenbrinkGKM16} says\footnote{Actually the sum in the statement of \cite[Lemma~2.1]{DBLP:conf/icalp/BerenbrinkGKM16} is over $v \in V$ rather than $v \in S_t$, a stronger bound, but this appears to be a typo as can be seen by comparing it to the last line of their proof, so we cite the weaker version for safety. There are no wider implications for the correctness of~\cite{DBLP:conf/icalp/BerenbrinkGKM16}, as they bound by the sum over $S_t$ everywhere they use Lemma~2.1.} that
  \[
    \E[\psi(S_{t+1}) \, \lvert \,  S_t=s_t] \le \psi(s_t) - \sum_{u
    \in S_t}\frac{\lambda_{u,t}d(u)}{32\psi(s_t)^3} \le \psi(s_t) -
    \frac{\dmin}{32\psi(s_t)^3}\sum_{u \in S_t}\lambda_{u,t}.
  \]
  Observe that each edge between $s_t$ and $\overline{s_t}$ is
  counted exactly once in this sum, so $\sum_{u \in V}\lambda_{u,t}
  = e_t(s_t,\overline{s_t})$. 
  Moreover, by symmetry of opinions, the conditioning on the value of $S_t$ in this statement is equivalent to conditioning on the value of $(A_t,\overline{A_t})$; since the voter model is
  Markov, this in turn is equivalent to conditioning on the value of 
  $\filtration{t}$, so the result follows.
\end{proof}

It is not hard to turn \cref{lem:potdec-regular} into a bound on the backward bias of $\psi(S_t)$ over a whole interval $t \in [T_j,T_{j+1})$,
using the fact that $\Vol(S_t)$ is approximately constant over $t \in
[T_j,T_{j+1})$ by the definition of $T_{j+1}$.

\begin{corollary}\label{cor:potdec-regular}\leanok%
  Let $j,t_j \ge 0$, and suppose $\filtrationValue{t_j}$ is a possible value
  of $\filtration{T_j}$. Let $s_{t_j}$ be the value of
  $S_{T_j}$ determined by $\filtrationValue{t_j}$. Let
  \[
    \mu \coloneqq \E\Big[\sum_{t=t_j}^{T_{j+1}-1}e_t(S_t,\overline{S_t})
    \,\Big|\, \filtration{T_j}=\filtrationValue{t_j}\Big].
  \]
  Then
  \[
    \E[\psi(S_{T_{j+1}}) - \psi(s_{t_j})\mid \filtration{T_j}=\filtrationValue{t_j}]
    \le
    -\frac{\dmin}{24\sqrt{6}\psi(s_{t_j})^3}\mu.
  \]
\end{corollary}
\begin{proof}
  Observe that if $s_{t_j} = \emptyset$ then $\mu=0$, $T_{j+1} = t_j$, and the result is immediate. Therefore suppose for the rest of the proof that $s_{t_j} \ne \emptyset$.

  For brevity, let $M \coloneqq \E[\psi(S_{T_{j+1}}) - \psi(s_{t_j})\mid
  \filtration{T_j}=\filtrationValue{t_j}]$. Decomposing the sum into time steps, we have
  \begin{align*}
    M &=
    \E\Big[\sum_{t=t_j}^{T_{j+1}-1}(\psi(S_{t+1})-\psi(S_t))\,\Big|\,\filtration{T_j}=\filtrationValue{t_j}\Big]\\
    &= \sum_{t=t_j}^{I_j^+}
    \E\big[1_{t<T_{j+1}}(\psi(S_{t+1})-\psi(S_t))
    \,\big|\,\filtration{T_j}=\filtrationValue{t_j}\big].
  \end{align*}
  Applying the tower property to expectations over each $\filtration{t}$ (subject to
  the existing conditioning on $\filtration{T_j}$) yields
  \begin{align}\label{eq:cor-pdr}%
    M &= \sum_{t=t_j}^{I_j^+}
    \E\Big[\E\big[1_{t<T_{j+1}}(\psi(S_{t+1})-\psi(S_t))\,\big|\,\filtration{t}\big]\,\Big|\,\filtration{T_j}=\filtrationValue{t_j}\Big].
  \end{align}

  Let $\filtrationValue{t}$ be a possible value of $\filtration{t}$ extending $\filtrationValue{t_j}$. If
  $\filtrationValue{t}$ entails $t \ge T_{j+1}$ then
  \[
    \E[1_{t<T_{j+1}}(\psi(S_{t+1})-\psi(S_t)) \mid \filtration{t}] = 0.
  \]
  Otherwise, $t < T_{j+1} \le T_{j_{\min}}$, so $\filtrationValue{t}$ determines a
  value $s_t$ of $S_t$ with $\Vol(s_t) \in [\tfrac12\Vol(s_{t_j}),
  \tfrac32\Vol(s_{t_j})]$ and in particular with $s_t \ne \emptyset$. Thus
  by \cref{lem:potdec-regular},
  \[
    \E\big[1_{t<T_{j+1}}(\psi(S_{t+1})-\psi(S_t))\,\big|\,\filtration{t}=\filtrationValue{t}\big]
    \le -\frac{\dmin}{32}\cdot
    \frac{e_t(s_t,\overline{s_t})}{\psi(s_t)^3} \le
    -\frac{\dmin}{32(3/2)^{3/2}}
    \cdot\frac{e_t(s_t,\overline{s_t})}{\psi(s_{t_j})^3}.
  \]
  Observe that $32(3/2)^{3/2} = 24\sqrt{6}$, so combining the two cases yields
  \[
    \E\big[1_{t<T_{j+1}}(\psi(S_{t+1})-\psi(S_t))\,\big|\,\filtration{t}\big]
    \le
    -1_{t<T_{j+1}}\frac{\dmin}{24\sqrt{6}}\cdot\frac{e_t(S_t,\overline{S_t})}{\psi(s_{t_j})^3}.
  \]
  Substituting back into~\eqref{eq:cor-pdr} then gives
  \begin{equation*}
    M \le
    -\frac{\dmin}{24\sqrt{6}\psi(s_{t_j})^3}\sum_{t=t_j}^{I_j^+}
    \E\big[1_{t<T_{j+1}}e_t(S_t,\overline{S_t}) \,\big|\,\filtration{T_j}=\filtrationValue{t_j}\big],
  \end{equation*}
  and so the result follows.
\end{proof}

In order to apply \cref{cor:potdec-regular} we will need to bound $\E[\sum_{t=t_j}^{T_{j+1}-1} e_t(S_t,\overline{S_t}) \mid \calH_{T_j}=H_{t_j}]$ below, which we do in \cref{lem:prob-good-event}. Informally, the proof will proceed as follows. Write $s_{t_j}$ for the value of $S_{t_j}$ determined by $\calH_{T_j}=H_{t_j}$. Since $I_j \subseteq [T_j,I_j^+]$, there must be ``one good step'' $t' \in [T_j, I_j^+]$ with
$\phi^{t'}(s_{t_j}) \ge \phi_j$ and hence $e_{t'}(s_{t_j},\overline{s_{t_j}}) \ge \phi_j\Vol(s_{t_j})$. This is promisingly close to a lower bound on $e_{t'}(S_{t'},\overline{S_{t'}})$, which would give us the required lower bound on the mean assuming $t' < T_{j+1}$, but not close enough --- $S_{t'}$ may be very far from $s_{t_j}$. Since $[T_j,T_{j+1})$ is stable, we expect this to be unlikely, but we must formalise this intuition. We can decompose
\begin{align*}
&\E\Big[\sum_{t=t_j}^{T_{j+1}-1} e_t(S_t,\overline{S_t}) \mid \calH_{T_j}=H_{t_j}\Big] \ge\\
&\qquad\E\Big[\sum_{t=t_j}^{(T_{j+1}\wedge t')-1} e_t(S_t,\overline{S_t}) \mid \calH_{T_j}=H_{t_j}] + \E[1_{t'<T_{j+1}}e_{t'}(S_{t'},\overline{S_{t'}}) \mid \calH_{T_j}=H_{t_j}].
\end{align*}
If the first term is small, then we will show using \cref{lem:stepwise-edges-bound} that $e_{t'}(S_{t'},\overline{S_{t'}})$ \emph{is} close to $e_{t'}(s_{t_j},\overline{s_{t_j}})$, and so the second term must be large.

\begin{lemma}\label{lem:prob-good-event}\leanok%
  Let $j,t_j \ge 0$, and suppose $\filtrationValue{t_j}$ is a stable possible value
  of $\filtration{T_j}$.
  Let $s_{t_j} \ne \emptyset$ be the value of $S_{T_j}$ determined by $\filtration{T_j}=\filtrationValue{t_j}$.
  Then
  \begin{equation}\label{eq:prob-good-event-statement}
    \mu\coloneqq \E\Big[\sum_{t=t_j}^{T_{j+1}-1}e_t(S_t,\overline{S_t})\,\Big|\,\calH_{T_j}=H_{t_j}\Big] \ge \phi_j\Vol(s_{t_j})/8.
  \end{equation}
\end{lemma}

\begin{proof}
  Let $\filtrationValue{t_j}\in \filtration{T_j}$ be a stable possible value of
  $\filtration{T_j}$. By the definition of $T_{j}$
  we know that $t_{j} \le I_{j-1}^++1 = I_j^-$. By the hypothesis of \cref{thm:upper-bound-window-conductance}, $\phi^{I_j}(\calG) \ge \phi_j$, so there exists $t'\in I_j \subseteq [t_j, I_j^+]$ with
  \[
    e_{t'}(s_{t_j},\overline{s_{t_j}}) \ge \phi_j\Vol(s_{t_j}).
  \]
  Let
  \[
    \mu' \coloneqq \E\Big[\sum_{t=t_j}^{(T_{j+1}\wedge t')-1}e_t(S_t,\overline{S_t})\,\Big|\,\calH_{T_j}=H_{t_j}\Big].
  \]
  It is immediate that $\mu \ge \mu'$, so if $\mu' \ge \phi_j\Vol(s_{t_j})/8$ then \cref{eq:prob-good-event-statement} holds and we are done. For the rest of the proof, suppose instead $\mu' < \phi_j\Vol(s_{t_j})/8$.
  
  We will bound $\mu$ by its term at $t=t'$. Observe that
  \begin{align*}
    \mu &= \sum_{t=t_j}^{I_j^+} \E\big[1_{t<T_{j+1}} e_t(S_t,\overline{S_t}) \,\big|\, \calH_{T_j} = H_{t_j}\big]
    \ge \E\big[1_{t'<T_{j+1}} e_{t'}(S_{t'},\overline{S_{t'}})\mid \calH_{T_j}=H_{t_j}\big].
  \end{align*}
  Let $\calE$ be the event that $T_{j+1} > t'$ and $e_{t'}(S_{t'},\overline{S_{t'}}) \ge \phi_j \Vol(s_{t_j})/4$. Then it follows that
  \begin{align}\label{eq:prob-good-event-E}
      \mu \ge \pr(\calE \mid \calH_{T_j} = H_{t_j}) \cdot \phi_j\Vol(s_{t_j})/4.
  \end{align}
  We break $\calE$ into two sub-events: $\calE_1$, the event that $T_{j+1} > t'$; and $\calE_2$, the event that $e_{t'}(S_{t'},\overline{S_{t'}}) \ge \phi_j\Vol(s_{t_j})/4$. Then
  \begin{align}
    \pr\left(\calE\mid
    \filtration{T_j}=\filtrationValue{t_j}\right)=1-\pr\left(\ov{\calE_1}\mid
    \filtration{T_j}=\filtrationValue{t_j}\right)-\pr\left(\calE_1\cap\ov{\calE_2}\mid
    \filtration{T_j}=\filtrationValue{t_j}\right).\label{eq:prob-good-event-decompb}
  \end{align}

  In order to apply \cref{eq:prob-good-event-E}, we bound each term of \cref{eq:prob-good-event-decompb} separately, starting with the first. By the definition
  of $T_{j+1}$,
  \begin{equation*}
    \{T_{j+1}\leq t'\}\subseteq\{T_{j+1} \le I_j^+\} \subseteq
    \big\{\abs{\Vol(S_{T_{j+1}}) - \Vol(s_{t_j})}\ge \tfrac12\Vol(s_{t_j})\big\}.
  \end{equation*}
  Therefore,
  \begin{equation*}
    \pr(\ov{\calE_1} \mid \filtration{T_j}=\filtrationValue{t_j})\leq
    \pr\big(\abs{\Vol(S_{T_{j+1}}) - \Vol(s_{t_j})}\ge
    \tfrac12\Vol(s_{t_j})\, \big| \,  \filtration{T_j} =\filtrationValue{t_j}\big).
  \end{equation*}
  By Markov's inequality, it follows that
  \begin{equation*}
    \pr(\ov{\calE_1} \mid \filtration{T_j}=\filtrationValue{t_j})\leq
    \frac{\E\big[\abs{\Vol(S_{T_{j+1}}) - \Vol(s_{t_j})} \mid
    \filtration{T_j}=\filtrationValue{t_j}\big]}{\tfrac12\Vol(s_{t_j})}.
  \end{equation*}
  Since $\filtrationValue{t_j}$ is stable, $\E[\abs{\Vol(S_{T_{j+1}}) -
  \Vol(s_{t_j})} \mid \filtration{T_j}=\filtrationValue{t_j}]\leq \Vol(s_{t_j})/8$, so
  \begin{equation*}
    \pr(\ov{\calE_1} \mid \filtration{T_j}=\filtrationValue{t_j})\leq 1/4.
  \end{equation*}
  For brevity, let $p \coloneqq \pr(\calE_1\cap\ov{\calE_2} \mid
  \filtration{T_j}=\filtrationValue{t_j})$; then by~\cref{eq:prob-good-event-decompb} it
  follows that
  \begin{equation}\label{eq:prob-good-event-decomp-2b}%
    \pr(\calE \mid \filtration{T_j}=\filtrationValue{t_j}) \ge 3/4 - p\,.
  \end{equation}

  We now bound $p$ above. Since $\ov{\calE_2}$ and
  $e_{t'}(s_{t_j},\ov{s_{t_j}}) \ge \phi_j\Vol(s_{t_j})$
  together imply
  $\abs{e_{t'}(S_{t'},\ov{S_{t'}})-e_{t'}(s_{t_j},\ov{s_{t_j}})}\geq
  \tfrac{3}{4}\phi_j\Vol(s_{t_j})$,
  \begin{align*}
    p&\leq\pr\Bigl(
      \{T_{j+1}>t'\}\cap\big\{\abs{e_{t'}(S_{t'},\ov{S_{t'}})-e_{t'}(s_{t_j},\ov{s_{t_j}})}\geq
        \tfrac{3}{4}
    \phi_j\Vol(s_{t_j})\big\} \Bigm| \filtration{T_j}=\filtrationValue{t_j}\Bigr)
    \\&=\pr\Bigl(1_{T_{j+1}>t'}\cdot\abs{e_{t'}(S_{t'},\ov{S_{t'}})-e_{t'}(s_{t_j},\ov{s_{t_j}})}\geq
      \tfrac{3}{4}\phi_j\Vol(s_{t_j}) \Bigm| \filtration{T_j}=\filtrationValue{t_j}\Bigr).
  \end{align*}
  By Markov's inequality, letting
  \[
  \nu
  \coloneqq
  \E\Bigl[1_{T_{j+1}>t'}\abs{e_{t'}(S_{t'},\ov{S_{t'}})-e_{t'}(s_{t_j},\ov{s_{t_j}})}\Bigm|
  \filtration{T_j}=\filtrationValue{t_j}\Bigr]\,,
  \]
  it follows that
  \begin{align}\label{eq:markov-event2}%
    p\leq \frac{\nu}{\tfrac{3}{4}\phi_j\Vol(s_{t_j})}.
  \end{align}

  We next bound $\nu$ above. By the triangle inequality applied to a telescoping~sum,
  \begin{equation*}
    \nu\leq
    \E\Big[1_{T_{j+1}>t'}\sum_{t=t_j}^{t'-1}\abs{e_{t'}(S_{t+1},\ov{S_{t+1}})-e_{t'}(S_{t},\ov{S_{t}})}\,\Big|\,\filtration{T_j}=\filtrationValue{t_j}\Big].
  \end{equation*}
  By the tower property, taking
  expectations over $\filtration{t}$ subject to our existing conditioning on
  $\filtration{T_j}$, it follows that
  \begin{align*}
    \nu &\le
    \sum_{t=t_j}^{t'-1}\E\Bigl[\E\bigl[1_{T_{j+1}>t}\cdot\abs{e_{t'}(S_{t+1},\ov{S_{t+1}})-e_{t'}(S_{t},\ov{S_{t}})}\bigm|
    \filtration{t}\bigr] \Bigm| \filtration{T_j}=\filtrationValue{t_j}\Bigr]\\
    &=\sum_{t=t_j}^{t'-1}\E\Bigl[1_{T_{j+1}>t}\cdot \E\bigl[\abs{e_{t'}(S_{t+1},\ov{S_{t+1}})-e_{t'}(S_{t},\ov{S_{t}})}\bigm|
    \filtration{t}\bigr] \Bigm| \filtration{T_j}=\filtrationValue{t_j}\Bigr].
  \end{align*}
  By \cref{lem:stepwise-edges-bound} applied with that lemma's $j$ and $t$ equal to our $t$ and $t'$, respectively, it follows that
  \begin{align*}
    \nu &\le \sum_{t=t_j}^{t'-1}\E\bigl[1_{T_{j+1}>t}\cdot
    e_t(S_t,\overline{S_t}) \bigm| \filtration{T_j}=\filtrationValue{t_j}\bigr]\\
    &= \E\Bigl[\,\sum_{t=t_j}^{(T_{j+1}\wedge t')-1}e_t(S_t,\overline{S_t})
    \Bigm|\filtration{T_j}=\filtrationValue{t_j}\,\Bigr] = \mu' \le
    \phi_j\Vol(s_{t_j})/8\,.
  \end{align*}
  Plugging this bound on $\nu$ into~\cref{eq:markov-event2}, we
  obtain $p \le 1/6$. By \cref{eq:prob-good-event-decomp-2b} it follows that $\pr(\calE \mid \calH_{T_j} = H_{t_j}) > 1/2$, so the result follows from \cref{eq:prob-good-event-E}.
\end{proof}

We now prove the main result of the subsection.

\begin{lemma}\label{lem:psi-down-drift}\leanok%
    Let $j,t_j \ge 0$, and suppose $H_{t_j}$ is a possible value of $\calH_{T_j}$. Let $s_{t_j}$ be the value of $S_{T_j}$ determined by $\calH_{T_j} = H_{t_j}$. Then
    \[
        \E\big[\psi(S_{T_{j+1}})-\psi(s_{t_j}) \,\big|\, \calH_{T_j}=H_{t_j}\big] \le 0.
    \]
    If in addition $s_{t_j} \ne \emptyset$ and $H_{t_j}$ is stable, then
    \[
        \E\big[\psi(S_{T_{j+1}})-\psi(s_{t_j}) \,\big|\, \calH_{T_j}=H_{t_j}\big] \le -\frac{d_{\min}\phi_j}{500\psi(s_{t_j})}.
    \]
\end{lemma}
\begin{proof}
  By \cref{cor:potdec-regular}, taking
  \[
    \mu\coloneqq \E\Big[\sum_{t=t_j}^{T_{j+1}-1}e_t(S_t,\overline{S_t})
    \,\Big|\, \filtration{T_j}=\filtrationValue{t_j}\Big],
  \]
  we have
  \begin{align*}
    \E\big[\psi(S_{T_{j+1}}) - \psi(s_{t_j})\,\big|\, \filtration{T_j}=\filtrationValue{t_j}\big]
    &\le
    -\frac{\dmin}{24\sqrt{6}\psi(s_{t_j})^3}\mu.
  \end{align*}
  Since $\mu \ge 0$, the first part of the result follows. 
  If $s_{t_j} \ne \emptyset$ and $H_{t_j}$ is stable, then by \cref{lem:prob-good-event} it further follows that
  \begin{align*}
    \E\big[\psi(S_{T_{j+1}}) - \psi(s_{t_j})\,\big|\, \filtration{T_j}=\filtrationValue{t_j}\big] \le -\frac{\dmin\phi_j\Vol(s_{t_j})}{8\cdot 24\sqrt{6}\psi(s_{t_j})^3} = -\frac{\dmin\phi_j}{192\sqrt{6}\psi(s_{t_j})},
  \end{align*}
  as required.
\end{proof}

%% file: sections/3.2_case2.tex
\subsection{Case 2: Unstable Intervals}\label{sec:case-2}

\begin{lemma}\label{lem:chi-down-drift-generic}\leanok%
  Suppose that $(X_t)$ is a non-negative
  integer-valued supermartingale with filtration $(\sigma(\genericfiltration{t}) \colon t \geq0)$ and initial state $x_0\coloneqq X_0$. Fix $t \ge 0$, let $F_t$ be a possible value of $\calF_t$, let $x_t$ be the value of $X_t$ determined by $F_t$, and suppose $\xi \ge 0$ satisfies
  \[
  \E\bigl[\abs{X_{t+1} - x_t}\bigm|\genericfiltration{t}=\genericValue{t}\bigr] \ge \xi x_t.
  \]
  Then 
  \[
  \E[\log(1+X_{t+1}) - \log(1+X_t)\mid \genericfiltration{t}=\genericValue{t}] \le -1_{x_t>0}\xi^2/96.
  \]
\end{lemma}
\begin{proof}
    For brevity, write $\zeta(x) \coloneqq \log(1+x)$ for all $x\ge 0$. If $x_t = 0$ then $X_{t+1}$ must equal zero (since $X_{t+1} \ge 0$ and $\E[X_{t+1} \mid \calF_t=F_t] \le x_t = 0$); thus the result holds. We may therefore assume $x_t>0$, and hence $x_t \ge 1$, for the rest of the proof.
    
    Let $\Delta_t \coloneqq X_{t+1}-x_t$, and observe that conditioned on $\calF_t=F_t$,
    \[
        \zeta(X_{t+1})-\zeta(x_t) = \log\Big(\frac{1+X_{t+1}}{1+x_t}\Big) = \log\Big(1+\frac{\Delta_t}{1+x_t}\Big).
    \]
    In this expression, $\Delta_t/(1+x_t) \ge -x_t/(1+x_t) > -1$. By the Taylor expansion, for all $z\in (-1,0)$ we have $\log(1+z) \le z-z^2/2+z^3/3 \le z-z^2/6$; moreover, for all $z \ge 0$ we have $\log(1+z) \le z$. 
    Thus taking $z\coloneqq\Delta_t/(1+x_t)$, we obtain
    \[
        \zeta(X_{t+1})-\zeta(x_t)=\log(1+z) \le \frac{\Delta_t}{1+x_t} - \frac{1_{\Delta_t<0}\Delta_t^2}{6(1+x_t)^2}.
    \]
     Since $(X_t)$ is a supermartingale, we have $\E[\Delta_t \mid \calF_t=F_t] \le 0$, so on taking expectations we obtain
    \begin{equation*}
        \E[\zeta(X_{t+1})-\zeta(x_t) \mid \calF_t=F_t] \le -\frac{\E\big[1_{\Delta_t<0}\Delta_t^2\mid \calF_t=F_t\big]}{6(1+x_t)^2}.
    \end{equation*}
    Since $x_t \ge 1$, we have $1+x_t \le 2x_t$; moreover, for all real variables $Z$ we have $\E[Z^2] \ge \E[Z]^2$, and $1_{\Delta_t<0}\Delta_t^2 = (1_{\Delta_t<0}\Delta_t)^2$. It follows that
    \begin{equation}\label{eq:new-jump-chain-foo}
        \E[\zeta(X_{t+1})-\zeta(x_t) \mid \calF_t=F_t] \le -\frac{\E\big[1_{\Delta_t<0}\Delta_t\mid \calF_t=F_t\big]^2}{24x_t^2}.
    \end{equation}

    We next bound this expectation by splitting $\Delta_t$ into positive and negative parts. Let $\Delta_t^+ \coloneqq \Delta_t \vee 0$ and $\Delta_t^- \coloneqq -(\Delta_t \wedge 0)$, so that: $\Delta_t^+, \Delta_t^- \ge 0$; $\Delta_t = \Delta_t^+ - \Delta_t^-$; $|\Delta_t| = \Delta_t^+ + \Delta_t^-$; and $1_{\Delta_t<0}\Delta_t = -\Delta_t^-$.
    Rewriting~\cref{eq:new-jump-chain-foo} in this notation, we obtain
    \begin{equation}\label{eq:new-jump-chain-bar}
        \E[\zeta(X_{t+1})-\zeta(x_t) \mid \calF_t=F_t] \le -\frac{\E\big[\Delta_t^-\mid \calF_t=F_t\big]^2}{24x_t^2}.
    \end{equation}
    Since $(X_t)$ is a supermartingale, 
    \[
    \E[\Delta_t^+ \mid \calF_t=F_t] - \E[\Delta_t^- \mid \calF_t=F_t] = \E[\Delta_t \mid \calF_t = F_t] \le 0,
    \]
    and so $\E[\Delta_t^- \mid \calF_t=F_t] \ge \E[\Delta_t^+ \mid \calF_t=F_t]$. Thus by hypothesis,
    \[
        \E[\Delta_t^- \mid \calF_t = F_t] \ge \tfrac12 \E[\Delta_t^- + \Delta_t^+ \mid \calF_t = F_t] = \tfrac12 \E[|\Delta_t| \mid \calF_t=F_t]\ge \xi x_t/2.
    \]
    Plugging this into \cref{eq:new-jump-chain-bar} then yields
    \begin{equation*}
        \E[\zeta(X_{t+1})-\zeta(x_t) \mid \calF_t=F_t] \le -\frac{\xi^2x_t^2}{96x_t^2} = -\frac{\xi^2}{96}.
    \end{equation*}
    as required.
\end{proof}

\begin{definition}\label{def:chi-potential}\leanok%
    For all sets $S \subseteq V(\calG)$, let $\chi(S) = \log(1+\Vol(S))$.
\end{definition}

\begin{corollary}\label{cor:chi-down-drift-voter}\leanok%
    Let $j,t_j \ge 0$, and suppose $H_{t_j}$ is a possible value of $\calH_{T_j}$. Let $s_{t_j}$ be the value of $S_{T_j}$ determined by $\calH_{T_j} = H_{t_j}$. Then
    \[
        \E\big[\chi(S_{T_{j+1}})-\chi(s_{t_j}) \,\big|\, \calH_{T_j}=H_{t_j}\big] \le 0.
    \]
    If in addition $s_{t_j} \ne \emptyset$ and $H_{t_j}$ is unstable, then
    \[
        \E\big[\chi(\Vol(S_{T_{j+1}}))-\chi(\Vol(s_{t_j})) \,\big|\, \calH_{T_j}=H_{t_j}\big] \le -10^{-4}.
    \]
\end{corollary}
\begin{proof}
    We apply \cref{lem:chi-down-drift-generic}, taking $(X_j)$ to be the embedded volume process $\Vol(S_{T_j})$. For the first part of the statement we take $\xi=0$, and for the second part we take $\xi=1/8$ (which we can do by the definition of an unstable filtration).
\end{proof}

%% file: sections/3.3_interplay.tex
\subsection{Interplay Between the Two Cases}\label{sec:combining-cases}%

\UpperBoundTwoOpinion*
\begin{proof}
    If $s_0 = \emptyset$ then the result is immediate, so suppose $s_0 \ne \emptyset$. Let 
    \begin{align*}
    \tau &\coloneqq \min\{t \ge 0\colon S_t=\emptyset\} \wedge T_J,\\
    \tau' &\coloneqq \tau \wedge \min\{t \ge 0\colon \Vol(S_t)\ge 8\Vol(s_0)\}.
    \end{align*}
    Let $\calE$ be the event that either 
     $\Vol(S_{\tau'}) \ge 8\Vol(s_0)$ or $\tau' = T_J$ and $S_{\tau'} \ne \emptyset$.
    We will first prove that $\pr(\Vol(S_{\tau'}) \ge 8\Vol(s_0)) \le 1/8$, then prove that $\pr(\tau' = T_J \text{ and }S_{\tau'} \ne \emptyset) \le 3/8$. Given these two facts, the result will follow easily by a union bound.

    \medskip\noindent\textbf{Bounding $\boldsymbol{\pr(\Vol(S_{\tau'})\geq 8\Vol(s_0))}$:}
    By Markov's inequality,
    \begin{equation}\label{eq:final-two-opinion-markov}
        \pr(\Vol(S_{\tau'})\geq 8\Vol(s_o))\le \frac{\E[\Vol(S_{\tau'})]}{8\Vol(s_0)}.
    \end{equation}
    Recall that $(\Vol(S_t))$ is a supermartingale, and $\tau'$ is a finite stopping time for it. Thus by the optional stopping theorem, $\E(\Vol(S_{\tau'})) \le \Vol(s_0)$. It therefore follows by \cref{eq:final-two-opinion-markov} that
    \begin{equation}\label{eq:final-two-opinion-event-1}
        \pr(\Vol(S_{\tau'})\geq 8\Vol(s_o)) \le 1/8,
    \end{equation}
    as claimed.

    \medskip\noindent\textbf{Bounding $\boldsymbol{\pr(\tau'=T_J \text{ and }S_{\tau'}\neq \emptyset)}$:} We will first establish downward drift of the following potential. For all $j \ge 0$, let 
    \[
        \Psi_j \coloneqq \frac{\psi(s_0)}{d_{\min}}\psi(S_{T_{j}}) + \chi(S_{T_{j}}).
    \]
    Temporarily fix $j \ge 0$, and suppose $H_{t_j}$ is a possible value of $\calH_{T_j}$. Let $s_{t_j}$ be the corresponding value of $S_{T_j}$ determined by $H_{t_j}$. Then we split into three cases, bounding the drift of $\Psi_j$ in each case. First suppose $H_{t_j}$ is stable and determines $t_j < \tau'$. This implies $s_{t_j} \ne \emptyset$, so by \cref{lem:psi-down-drift} and \cref{cor:chi-down-drift-voter},
    \begin{align*}
        \E\big[\Psi_{j+1}-\Psi_j\mid\calH_{T_j}=H_{t_j}\big] &= \frac{\psi(s_0)}{d_{\min}}\E\big[\psi(S_{T_{j+1}})-\psi(s_{t_j})\mid\calH_{T_j}=H_{t_j}\big]\ +\\
        &\qquad\qquad\qquad\qquad \E\big[\chi(S_{T_{j+1}})-\chi(s_{t_j})\mid\calH_{T_j}=H_{t_j}\big]\\
        &\le -\frac{\phi_j\psi(s_0)}{500\psi(s_{t_j})} + 0.
    \end{align*}
    Since $t_j < \tau'$, $\psi(s_{t_j}) \le \sqrt{8}\psi(s_0)$, so it follows that
    \[
        \E\big[\Psi_{j+1}-\Psi_j\mid\calH_{T_j}=H_{t_j}\big] \le -\phi_j/2000.
    \]
    Similarly, if instead $H_{t_j}$ is unstable and determines $t_j < \tau'$, then $s_{t_j} \ne \emptyset$, so by \cref{lem:psi-down-drift} and \cref{cor:chi-down-drift-voter},
    \[
        \E\big[\Psi_{j+1}-\Psi_j\mid\calH_{T_j}=H_{t_j}\big] \le 0-10^{-4} = -10^{-4}.
    \]
    Finally, if instead $H_{t_j}$ determines $t_j \ge \tau'$, then by \cref{lem:psi-down-drift} and \cref{cor:chi-down-drift-voter},
    \[
        \E\big[\Psi_{j+1}-\Psi_j\mid\calH_{T_j}=H_{t_j}\big] \le 0 + 0 = 0.
    \]
    Thus on combining all three cases, we obtain that for all $j \ge 0$,
    \begin{equation}\label{eq:final-two-opinion-drift}
        \E\big[\Psi_{j+1}-\Psi_j\mid\calH_{T_j}\big] \le -D_j,\mbox{ where } D_j\coloneqq 1_{T_j<\tau'}\Big(\frac{\phi_j}{4000} \wedge 10^{-4}\Big).
    \end{equation}

    Let $J' = \min\{j \ge 0 \colon T_j \ge \tau'\}$. Since $s_0 \ne \emptyset$ we have $\tau' > 0$, so $J' > 0$ and $T_{J'-1} < \tau' \le T_{J'}$. We will next apply the optional stopping theorem in a standard way to argue that the total accumulation of expected drift up to time $T_{J'}$, i.e.\ $\sum_{\ell=0}^{J'-1} D_\ell$, is unlikely to be large; we will exploit this to bound $J'$ above, in order to bound $T_{J'}$ above, in order to finally bound $\tau'$ above. 
    
    For all $j \ge 0$, let 
    \[
    X_j \coloneqq \Psi_j + \sum_{\ell=0}^{j-1}D_\ell.
    \]
    By \cref{eq:final-two-opinion-drift}, $(X_j)$ is a supermartingale under the filtration generated by $(\calH_{T_j})$; moreover, $J'$ is a finite stopping time for $(X_j)$. Thus by the optional stopping theorem, $\E[X_{J'}] \le X_0$.
    Expanding out both sides, we obtain
    \[
        \E[\Psi_{J'}] + \E\Big[\sum_{\ell=0}^{J'-1}D_\ell\Big] \le \Psi_0.
    \]
    Since $\Psi_{J'} \ge 0$, it follows that $\E[\sum_{\ell=0}^{J'-1}D_\ell] \le \Psi_0$. Markov's inequality then~yields
    \begin{equation}\label{eq:final-two-opinion-D-bound}
        \pr\Big(\sum_{\ell=0}^{J'-1} D_\ell \ge 4\Psi_0\Big) \le 1/4.
    \end{equation}

    We will next use \cref{eq:final-two-opinion-D-bound} to prove $\pr(J' \ge J) \le 1/4$. For this, it is enough to prove that $\sum_{\ell=0}^{J'-1} D_\ell < 4\Psi_0$ implies $J' < J$; by the definition of $J$, this is equivalent to proving
    \begin{equation}\label{eq:final-two-opinion-intermediate-goal}
        \sum_{\ell=0}^{J'} \phi_\ell < b\Big(\frac{\Vol(s_0)}{d_{\min}} + \log(1+\Vol(s_0))\Big).
    \end{equation}
    To this end, suppose that $\sum_{\ell=0}^{J'-1} D_\ell < 4\Psi_0$. By the definition of $J'$, for all $\ell < J'$ we have $T_\ell < \tau'$; thus by the definition of the $D_\ell$'s,
    \[
    \sum_{\ell=0}^{J'-1}D_\ell = \sum_{\ell=0}^{J'-1} \Big(\frac{\phi_\ell}{4000} \wedge 10^{-4}\Big) \ge 10^{-4}\sum_{\ell=0}^{J'-1}\phi_\ell.
    \]
    Since $\sum_{\ell=0}^{J'-1} D_\ell < 4\Psi_0$ by assumption, it follows that
    \[
        \sum_{\ell=0}^{J'-1}\phi_j < 4\cdot 10^4\Psi_0 = 4\cdot 10^4\Big(\frac{\Vol(s_0)}{d_{\min}} + \log(1+\Vol(s_0))\Big);
    \]
    on taking $b \ge 8\cdot 10^4$ this implies \cref{eq:final-two-opinion-intermediate-goal}, and hence $J' < J$ as required. We have therefore proved
    \begin{equation}\label{eq:final-two-opinion-J-bound}
        \pr(J' \ge J) \le 1/4.
    \end{equation}

    We now use \cref{eq:final-two-opinion-J-bound} to prove $\pr(\tau' = T_J\text{ and }S_{\tau'} \ne \emptyset) \le 1/4$, our original goal. It suffices to prove that whenever $\tau' = T_J\text{ and }S_{\tau'} \ne \emptyset$ occurs, $J'\ge J$. By the definition of $J'$, we have $T_{J'} \ge \tau'=T_J > T_{J'-1}$, where the final inequality follows since $S_{\tau'} \ne \emptyset$. This implies $J' > J-1$ and hence $J' \ge J$. It is therefore immediate from \cref{eq:final-two-opinion-J-bound} that
    \begin{equation}\label{eq:final-two-opinion-event-2}
        \pr(\tau' = T_J\text{ and }S_{\tau'} \ne \emptyset) \le 1/4<3/8.
    \end{equation}

    \medskip\noindent\textbf{Putting everything together:}
    Combining \cref{eq:final-two-opinion-event-1} and \cref{eq:final-two-opinion-event-2} with a union bound, we see that with probability at least $1/2$, $\calE$ does not occur, it follows that either $\Vol(S_{\tau'})<8\Vol(s_o)$ or both $\tau' = T_J$ and $S_{\tau'} = \emptyset$. In both cases, $\tau'=\tau.$ By the definition of $\tau$, this is equivalent to saying that $\tau$ is the consensus time. Since by the definitions of $\tau$ and $T_J$ we have $\tau \le T_J \le I_J^+ + 1 = I_{J+1}^-$, the result follows.
\end{proof}

%% file: sections/3.4_multi_opinion.tex
\subsection{Deriving the multi-opinion bound from the two-opinion bound}\label{sec:multi-opinion}

The following argument to derive a multi-opinion bound from a two-opinion bound follows essentially the same method as Berenbrink, Giakkoupis, Kermarrec, and Mallmann-Trenn~\cite{DBLP:conf/icalp/BerenbrinkGKM16}; however, since the lengths of our ``phases'' are more complex, \cref{thm:voter-absorb-two-opinion} has an additive error term not present in the two-opinion result of \cite{DBLP:conf/icalp/BerenbrinkGKM16}, and the presentation in~\cite{DBLP:conf/icalp/BerenbrinkGKM16} is quite abbreviated, we give the argument in full detail.

\UpperBound*
\begin{proof}
Let $\kappa$ be the number of opinions present in the model. The result is immediate if $\kappa=1$, so suppose $\kappa \ge 2$. Let $B$ be a large constant, to be determined later. First, we divide time into phases of deterministic length. Let $r_{\max} \coloneqq \lfloor Bm/d_{\min} \rfloor$. Let $\ell_0 = t_0 = 0$, and for all integers $r$ with $1 \le r \le r_{\max}$, let 
\[
\ell_r \coloneqq \min\Big\{\ell\colon \sum_{j=0}^{\ell-1}\phi_j \ge r\Big\},\qquad t_r \coloneqq I_{\ell_r}^- = \Delta_0 + \dots + \Delta_{\ell_r-1}.
\]
For all $r \ge 0$, we call the interval $[t_r, t_{r+1}-1]$ the \emph{$r$'th phase}. Note that since each $\phi_j$ lies in $[0,1]$, we have $t_0 < t_1 < \dots < t_{r_{\max}} \le \Delta_0 + \dots + \Delta_J < \infty$. 

Following~\cite{DBLP:conf/icalp/BerenbrinkGKM16}, we now group these phases into ``metaphases'' of random lengths. For all $t \ge 0$, let $\mathcal{O}(t)$ be the set of opinions remaining in the process at time $t$. Let $R_0 = 0$, and for all $\alpha \ge 1$, let
\[
R_\alpha \coloneqq \min\Big(\{r_{\max}\} \cup \big\{r\colon |\mathcal{O}(t_r)| \le (5/6)^\alpha \kappa \vee 1\big\}\Big).
\]
Then the \emph{$\alpha$'th metaphase} is the interval $[t_{R_\alpha},t_{R_{\alpha+1}}-1]$. Thus ignoring the cap at $r_{\max}$ (past which no more phases are defined), before consensus the end of the $\alpha$'th metaphase occurs after however many phases it takes for the number of opinions to have dropped by a cumulative factor of $(5/6)^\alpha$, and after consensus all metaphases are empty.

Let $\beta \coloneqq \lceil \log_{6/5} \kappa \rceil$, and observe that $(5/6)^\beta \kappa \le 1$. Thus by definition, if $R_\beta < r_{\max}$ then $\abs{\mathcal{O}(t_{R_\beta})} = 1$ and so the consensus time is at most $t_{R_\beta} < t_{r_{\max}} \le \Delta_0 + \dots + \Delta_J$ as required. 
It therefore suffices to prove that 
\begin{equation}\label{eq:multi-opinion-goal}
    \pr(R_\beta < r_{\max}) \ge 1/2.
\end{equation}

Our next goal is to prove that taken individually, each metaphase ends in few phases with probability bounded away from zero. We will prove the following claim, where we begin in the middle of a metaphase to allow for repeated application later.

\medskip\noindent\textbf{Claim:} Let $r,\alpha \ge 0$, let $H_{t_r}$ be a possible value of $\calH_{t_r}$, and suppose that $\calH_{t_r}=H_{t_r}$ determines $R_\alpha \le r < R_{\alpha+1}$. Let $r_\alpha$ and $O_\alpha$ be the values of $R_\alpha$ and $\mathcal{O}(t_{r_\alpha})$ determined by $H_{t_r}$, and let 
\[
\xi_\alpha \coloneqq 1 + \Big\lceil b\Big(\frac{6m}{d_{\min}\abs{O_\alpha}} + \log\big(1+6m/|O_\alpha|\big)\Big)\Big\rceil.
\] 
Then $\pr(R_{\alpha+1} > r+\xi_\alpha \mid \calH_{t_r} = H_{t_r}) \le 2/3$.

\medskip\noindent\textbf{Proof of Claim:} If $r+\xi_\alpha \ge r_{\max}$ then $R_{\alpha+1} \le r_{\max} \le r+\xi_\alpha$ with certainty and the Claim holds immediately, so suppose for the rest of the proof that $r+\xi_\alpha < r_{\max}$. Thus 
\begin{equation}\label{eq:multi-opinion-00a}
    \pr(R_{\alpha+1} > r+\xi_\alpha \mid \calH_{t_r}=H_{t_r}) = \pr\big(|\mathcal{O}({t_{r+\xi_\alpha}})| > (5/6)^{\alpha+1}\kappa \vee 1\mid\calH_{t_r}=H_{t_r}\big).
\end{equation}
Next observe that conditioned on $\calH_{t_r}=H_{t_r}$, since $r_\alpha < R_{\alpha+1}$, 
\begin{equation}\label{eq:multi-opinion-00b}
    (5/6)^{\alpha}\kappa \ge |O_\alpha| > (5/6)^{\alpha+1}\kappa \vee 1.
\end{equation}
Applying $(5/6)^\alpha\kappa \ge |O_\alpha|$ in \cref{eq:multi-opinion-00a}, it follows that
\begin{equation}\label{eq:multi-opinion-0a}
    \pr(R_{\alpha+1} > r+\xi_\alpha \mid \calH_{t_r} = H_{t_r}) \le \pr\big(|\mathcal{O}({t_{r+\xi_\alpha}})| \ge \tfrac56|O_\alpha|\mid \calH_{t_r} = H_{t_r}\big).
\end{equation}
Henceforth we will be concerned with this event: that at least $5/6$ of the original $|O_\alpha|$ opinions present at the start of the $r$'th phase remain after $\xi_\alpha$ further phases have elapsed.

We next express the right-hand side of \cref{eq:multi-opinion-0a} in terms of indicator variables. For each opinion $q \in O_\alpha$, let $X_q$ be the indicator variable of the event that $q \in \mathcal{O}(t_{r+\xi_\alpha})$ and that $\mathcal{O}(t_{r+\xi_\alpha}) \ne \{q\}$ (i.e.\ that $q$ has neither vanished nor taken over $\calG$ within $\xi_\alpha$ phases starting from $t_r$). Then by \cref{eq:multi-opinion-0a},
\begin{equation}\label{eq:multi-opinion-0b}
    \pr(R_{\alpha+1} > r+\xi_\alpha \mid \calH_{t_r} = H_{t_r}) \le \pr\Big(\sum_{q \in O_\alpha} X_q > \tfrac{5}{6}\abs{O_\alpha} \,\Big|\, \calH_{t_r} = H_{t_r}\Big).
\end{equation}
(Indeed, if this sum is at most $\tfrac56|O_\alpha|$ then by time $t_{r+\xi_\alpha}$, either some opinion in $O_\alpha$ has taken over or at least $|O_\alpha|/6$ opinions have vanished; in either case, at least $|O_\alpha|/6$ opinions have vanished,
so $|\mathcal{O}(t_{r+\xi_\alpha})| \le 5|O_\alpha|/6 \le (5/6)^{\alpha+1}\kappa$ and hence $R_{\alpha+1} \le r+\xi_\alpha$.)

We next restrict the sum in \cref{eq:multi-opinion-0b} to a subset of opinions in $O_\alpha$ with low volume. For all $q \in O_\alpha$, let $a_q(t_r)$ be the set of vertices with opinion $q$ at time $t_r$ under $\calH_{t_r}=H_{t_r}$. We say that an opinion $q \in O(t_r)$ is \emph{small} if $\Vol(a_q(t_r)) \le 6m/\abs{O_\alpha}$. Let $O_{t_r}^-$ be the set of small opinions in $O(t_r)$, and observe~that 
\begin{equation}\label{eq:multi-opinion-small-1}
    \abs{O_{t_r}^-} \ge \abs{O({t_r})} - \abs{O_\alpha}/3;
\end{equation}
indeed, if this were not the case, then the remaining (at least) $\abs{O_\alpha}/3$ opinions would have total volume greater than $ (6m/|O_\alpha|)\cdot (|O_\alpha|/3)=\Vol(V(\calG))$ which is not possible. 
Moreover, if $q \in O_\alpha \setminus O(t_r)$ then $q$ has already died out by time $t_r$ and so $X_q = 0$.~Thus
\[
    \sum_{q \in O_\alpha}X_q  = \sum_{q \in O(t_r)} X_q \le \sum_{q \in O_{t_r}^-} X_q + \abs{O(t_r) \setminus O_{t_r}^-}.
\]
Thus whenever $\sum_{q \in O_\alpha} X_q > \tfrac56\abs{O_\alpha}$, we have
\[
\sum_{q \in O_{t_r}^-} X_q > \tfrac56\abs{O_\alpha} - \abs{O(t_r) \setminus O_{t_r}^-}.
\]
Thus by \cref{eq:multi-opinion-0b}, it follows that
\begin{align}\label{eq:multi-opinion-0c}
    \pr(R_{\alpha+1} > r+\xi_\alpha& \mid \calH_{t_r} = H_{t_r}) \nonumber\\&\le \pr\Big(\sum_{q \in O_{t_r}^-} X_q \ge \tfrac56\abs{O_\alpha} - \abs{O(t_r)} + \abs{O_{t_r}^-} \,\Big|\, \calH_{t_r} = H_{t_r} \Big).
\end{align}
By Markov's inequality and \cref{eq:multi-opinion-0c},
\begin{align}\label{eq:multi-opinion-1}
    &\pr(R_{\alpha+1} > r+\xi_\alpha \mid \calH_t = H_t) \le \frac{\sum_{q \in O_{t_r}^-} \pr(X_q=1\mid\calH_t=H_t)}{\tfrac56\abs{O_\alpha} - \abs{O(t_r)} + \abs{O_{t_r}^-}}.
\end{align}

To bound the right-hand side of \cref{eq:multi-opinion-1}, fix a small opinion $q \in O_{t_r}^-$. We consider a two-opinion standard voter model on the temporal graph $\calG'_{t'} \coloneqq \calG_{t'-t_r}$, where one opinion is $q$ with initial state $a_q(t_r)$; thus the other opinion ($\ov{q}$, say) is formed by grouping all vertices with opinions other than $q$ at time $t_r$ in $\calG$.
This model couples to our present $\kappa$-opinion voter model in the natural fashion, with opinion $q$ evolving identically between the two models subject to a time offset of $t_r$. Recall that by definition, $t_r = \Delta_0 + \dots + \Delta_{\ell_r-1}$. Let $b$ be as in \cref{thm:voter-absorb-two-opinion}, and let
\[
J' \coloneqq \min\Big\{j\colon \sum_{\ell=\ell_r}^j \phi_\ell \ge b\Big(\frac{\Vol(a_q(t_r))}{d_{\min}} + \log\big(1+\Vol(a_q(t_r))\big)\Big)\Big\};
\]
then by \cref{thm:voter-absorb-two-opinion}, with probability at least $1/2$, opinions $q$ and $\ov{q}$ achieve consensus in the two-opinion model within time $\Delta_{\ell_r} + \dots + \Delta_{J'}$.

We next bound $J'$ above. Note that
\begin{equation}\label{eq:multi-opinion-small-2}
\sum_{\ell=\ell_r}^{\ell_{r+\xi_\alpha}-1} \phi_\ell = \sum_{\ell=0}^{\ell_{r+\xi_\alpha}-1} \phi_\ell - \sum_{\ell=0}^{\ell_r-1} \phi_\ell.
\end{equation}
If $r>0$ (so that $\ell_r>0$), then applying the definition of the phases and the fact that $\phi_{\ell_r-1} \in [0,1]$, $\sum_{\ell=0}^{\ell_r-1}\phi_\ell = \sum_{\ell=0}^{\ell_r-2}\phi_\ell - \phi_{\ell_r-1} \le r+1$; otherwise, $\sum_{\ell=0}^{\ell_r-1}\phi_\ell = 0 < r+1$. In either case, again applying the definition of the phases, it follows from \cref{eq:multi-opinion-small-2} that
\begin{align*}
\sum_{\ell=\ell_r}^{\ell_{r+\xi_\alpha}-1} \phi_\ell &\ge (r+\xi_\alpha) - r - 1 = \xi_\alpha - 1 =\Big\lceil b\Big(\frac{6m}{d_{\min}\abs{O_\alpha}} + \log\big(1+6m/|O_\alpha|\big)\Big)\Big\rceil.
\end{align*}
Since opinion $q$ is small, it follows that
\[
\sum_{\ell=\ell_r}^{\ell_{r+\xi_\alpha}-1} \phi_\ell \ge b\Big(\frac{\Vol(a_q(t_r))}{d_{\min}} + \log\big(1+\Vol(a_q(t_r))\big)\Big),
\]
and hence that $J' \le \ell_{r+\xi_\alpha}-1$. Thus in our original $\kappa$-opinion model, with probability at least $1/2$, opinion $q$ either vanishes or takes over by time $t_r + \Delta_{\ell_r} + \dots + \Delta_{\ell_{r+\xi_\alpha}-1} = t_{r+\xi_\alpha}$; that is, with probability at least $1/2$, $X_q = 0$.

We have just shown that for all $q \in O_{t_r}^-$, $\pr(X_q=1\mid \calH_{t_r}=H_{t_r}) \le 1/2$. Substituting this into \cref{eq:multi-opinion-1} yields
\[
    \pr(R_{\alpha+1}>r+\xi_\alpha \mid \calH_{t_r} = H_{t_r}) \le 
    \frac{\tfrac12\abs{O_{t_r}^-}}{\tfrac56\abs{O_\alpha} - \abs{O(t_r)} + \abs{O_{t_r}^-}}.
\]
By the assumption of the claim we have $r < R_{\alpha+1}$, so $|O(t_r)| > (5/6)^{\alpha+1}\kappa$; by \cref{eq:multi-opinion-00b} it follows that $|O(t_r)| > \tfrac56|O_\alpha|$.
It follows that the right-hand side of the above equation is non-increasing in $\abs{O_{t_r}^-}$. Since $\abs{O_{t_r}^-} \ge \abs{O(t_r)}-\tfrac13\abs{O_\alpha}$ by \cref{eq:multi-opinion-small-1}, it follows that
\[
\pr(R_{\alpha+1}>r+\xi_\alpha \mid \calH_{t_r} = H_{t_r}) \le 
    \frac{\tfrac12(\abs{O(t_r)}-\tfrac13\abs{O_\alpha})}{\tfrac56\abs{O_\alpha} - \tfrac13\abs{O_\alpha}} \le \frac{\tfrac12(\abs{O_\alpha}-\tfrac13\abs{O_\alpha})}{\tfrac12\abs{O_\alpha}} = \frac{2}{3}.
\]
We have therefore proved the Claim.

\medskip\noindent\textbf{Proof of Theorem from Claim:} From the start of metaphase $\alpha$ at phase $r$ (say), the Claim applied to $\alpha$ and $r,r+\xi_\alpha,r+2\xi_\alpha,\dots$ implies that the number of times we have to wait $\xi_\alpha$ phases before the next metaphase is dominated above by a geometric variable with parameter $1/3$.
Thus the expected number of phases before the next metaphase satisfies
\[
    \E\bigl[R_{\alpha+1} - R_\alpha \bigm| \calH_{t_{R_\alpha}}\bigr] \le 3\xi_\alpha \le 6b\Big(\frac{6m}{d_{\min}\abs{O_\alpha}} + \log\big(1+6m/|O_\alpha|\big)\Big).
\]
We now bound this expectation above by exposing $\calH_{t_{R_\alpha}}$. Let $H_{t_{r_\alpha}}$ be a possible value of $\calH_{t_{R_\alpha}}$ which entails $R_{\alpha+1} > r_\alpha$ (or the expectation is zero conditioned on $\calH_{t_{R_\alpha}} = H_{t_{r_\alpha}}$). Then by the definition of $R_\alpha$, conditioned on $\calH_{t_{R_\alpha}} = H_{t_{r_\alpha}}$ we have $(5/6)^{\alpha+1}\kappa < \abs{\mathcal{O}(t_{r_\alpha})} \le (5/6)^\alpha \kappa$ exactly as in \cref{eq:multi-opinion-00b}. 
Thus by the tower property,
\begin{align}\nonumber
    \E\big[R_{\alpha+1} - R_\alpha] &= \E\big[\E[R_{\alpha+1} - R_\alpha \mid \calH_{t_{R_\alpha}}]\big]\\\label{eq:multi-opinion-2}
    &\le 6b\Big(\frac{6m}{(5/6)^{\alpha+1} d_{\min}\kappa} + \log\Big(1+\frac{6m}{(5/6)^{\alpha+1}\kappa}\Big)\Big).
\end{align}

We now exploit \cref{eq:multi-opinion-2} to bound $R_\beta$. Recall our final goal of \cref{eq:multi-opinion-goal}. By Markov's inequality,
\begin{equation}\label{eq:multi-opinion-3}
    \pr(R_\beta \ge r_{\max}) \le \frac{\E(R_\beta)}{r_{\max}},
\end{equation}
where by \cref{eq:multi-opinion-2}, 
\begin{align*}
\E(R_\beta) &= \sum_{\alpha=0}^{\beta-1} \E(R_{\alpha+1}-R_\alpha) \le \frac{36bm}{d_{\min}\kappa}\sum_{\alpha=0}^{\beta-1} \frac{1}{(5/6)^{\alpha+1}} + 6b\sum_{\alpha=0}^{\beta-1}\log\Big(1+\frac{6m}{(5/6)^{\alpha+1}\kappa}\Big).
\end{align*}
Reparameterising the left sum by $\gamma\coloneqq\beta-1-\alpha$ and bounding all terms of the right sum uniformly gives
\begin{align*}
\E(R_\beta) &\le \frac{36bm}{\kappa d_{\min}(5/6)^{\beta}}\sum_{\gamma=0}^{\beta-1} {(5/6)^{\gamma}} + 6b\beta\log\Big(1+\frac{6m}{(5/6)^\beta\kappa}\Big).
\end{align*}
By the definition of $\beta$, $(5/6)^\beta \ge (5/6)/\kappa > 1/(2\kappa)$, and so
\begin{align*}
\E(R_\beta) &\le \frac{72bm}{d_{\min}}\sum_{\gamma=0}^{\beta-1} (5/6)^{\gamma} + 6b\beta\log(1+12m).
\end{align*}
Observe that $\sum_{\gamma=0}^{\beta-1}(5/6)^\gamma < \sum_{\gamma=0}^\infty (5/6)^\gamma = 6$, $\beta \le \log_{6/5} \kappa + 1 \le 10 \log n$ (where we use $n\ge\kappa\ge 2$), and $\log(1+12m) \le \log(13n^2) \le 10\log n$.
Thus
\[
    \E(R_\beta)\le\frac{432bm}{d_{\min}} + 600b\log^2 n.
\]
Observe that $m/d_{\min} \ge n/2 = \omega(\log^2 n)$ as $n\to\infty$, so it follows that for a suitable constant $B'>0$,  $\E(R_\beta) \le B'm/d_{\min}$.
Thus by \cref{eq:multi-opinion-3},
\[
    \pr(R_\beta \ge r_{\max}) \le \frac{B'm/d_{\min}}{r_{\max}} = \frac{B'm/d_{\min}}{\lfloor Bm/d_{\min}\rfloor}.
\]
Take $B = 4B'$. Since no vertex in $\calG$ is isolated, $m/d_{\min} \ge n/2 \ge 1$, so $\lfloor Bm/d_{\min}\rfloor \ge Bm/(2d_{\min})$; thus $\pr(R_\beta \ge r_{\max}) \le 2B'/B = 1/2$. We therefore obtain our goal of \cref{eq:multi-opinion-goal}, as required.
\end{proof}

\UpperBoundf*
\begin{proof}
Let $\Delta'$ be such that $\phi(\calG) = \phi^{\Delta'}(\calG)/\Delta'$ (recalling that such a $\Delta'$ exists by the definition of $\phi(\calG)$), and let $\phi' \coloneqq \phi^{\Delta'}(\calG)$.
Let $\Delta_j=\Delta'$, and $\phi_j=\phi'$ for all $j\geq0$; we will apply \cref{thm:upper-bound-window-conductance} to $\Delta_0,\Delta_1,\dots$ and $\phi_0,\phi_1,\dots$. By the definition of $\Phi(\calG)$, for all $j\geq0$, $\phi^{I_j}\geq\phi_j=\phi'$ holds. Therefore, by \cref{thm:upper-bound-window-conductance}, with probability at least $1/2$ the consensus time is at most $J\Delta'$ with $J=\lceil bm/(d_{min}\phi')\rceil-1$. Thus, with probability at least $1/2$ the consensus time is at most
\[
J\Delta'\leq\frac{bm}{d_{min}\phi'}\Delta'=\frac{bm}{d_{\min}\Phi(\calG)}
\]
The result follows by taking $b'=b$.

\end{proof}

%% file: sections/4_lower_bound.tex
\section{Lower Bound on Absorption Time}

In this section, we prove \cref{thm:lower-bound-3regular}.
In order to do so, we construct a specific temporal graph where
consensus is reached only slowly. 
The idea is to construct a temporal graph $\calG$ such that with high probability, the opinion-zero set of $\calG$ sampled at appropriate stopping times is a blown-up contiguous interval on a length-$z$ cycle. The endpoints of this interval sampled at the same stopping times will then evolve independently according to an unbiased random walk with bounded step size; it is then not hard to show that this interval requires $\Omega(z^2)$ steps to reach the empty interval or the whole cycle, i.e.\ consensus.
Essentially, our temporal graph alternates between $G_0$ and $G_1$, which are disjoint unions of $z/2$ cliques on $2k$ vertices, arranged in a cycle.
See \cref{fig:lower-bound} for an illustration.

\begin{definition}\label{def:lower-bound-construction}\leanok%
	Let $T,k,z \ge 1$ be integers with $z$ even. Then the temporal graph $\calG^{T,k,z}$ is defined as follows.

	Let $V(\calG)$ be a disjoint union of sets $V_1,\dots,V_z$ of $k$ vertices each. For all $i \in [z]$, let $H_i$ be a $2k$-vertex clique spanning $V_i$ and $V_{i+1}$ (modulo $z$). Let
	\[
		G_0 = \bigcup_{i \in [z]\colon i\textrm{ is even}} H_i,\qquad G_1 = \bigcup_{i \in [z]\colon i\textrm{ is odd}} H_i.
	\]
	For all $j \ge 1$, let $I_j = \{(j-1)T,\dots,jT-1\}$. Then $\calG$ is the temporal graph which, throughout each time interval $I_j$, takes value $G_{j\textnormal{ mod } 2}$.
\end{definition}
\begin{remark}\label{rem:lower-bound-basic-properties}\leanok%
	Let $T,k,z \ge 1$ be integers with $z$ even. Then $\calG^{T,k,z}$ is $(2k-1)$-regular and has $n=kz$ vertices.
\end{remark}

We first bound the conductance of any interval of size $3T$ for any possible minority set $S\subseteq V(\calG^{T,k,z})$ below in \cref{lem:clique-lower-bound-conductance}. After that, we bound its absorption time in \cref{lem:lower-bound-absorption}, then choose specific values of $T$, $k$ and $z$ and derive \cref{thm:lower-bound-3regular} as an easy consequence.

\begin{figure}[ht]
	\centering
	\includegraphics[scale=0.42]{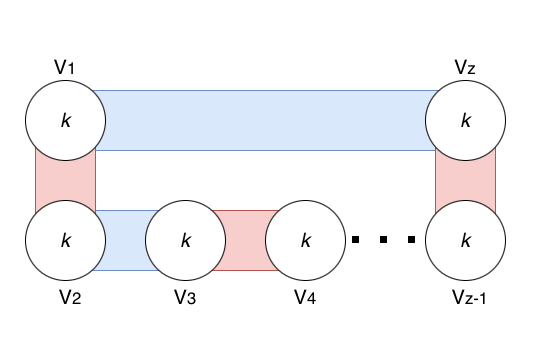}
	\caption{\label{fig:lower-bound}At even intervals the blue edges are active forming a $2k$-vertex clique. At odd intervals the red edges are active forming a $2k$-vertex clique.}
\end{figure}

In order to bound the conductance of any interval of size $3T$ for any half-volume set $S\subseteq V(\calG^{T,k,z})$ below in \cref{lem:clique-lower-bound-conductance}, we will need the following well-known (and trivial) bound.

\begin{lemma}\label{lem:static-clique-conductance}\leanok%
	For all $k \ge 1$, the complete graph on $2k$ vertices has conductance greater than $1/2$.
\end{lemma}

\begin{proof}
	Let $K_{2k}$ be the complete graph with~$2k$ vertices. For any set $S \subseteq V(K_{2k})$ with $0<\Vol(S) \le \Vol(K_{2k})/2$, we have $0<\abs{S} \le k$, and thus
	\[
		\phi(S)=\frac{e(S,\overline{S})}{\Vol(S)} = \frac{\abs{S}\cdot(2k-\abs{S})}{\abs{S}\cdot(2k-1)} = \frac{2k-\abs{S}}{2k-1}
		\geq \frac{k}{2k-1}>\frac12\,.
	\]
	By definition, $\phi(K_{2k}) = \min_{S}\phi(S) > 1/2$.
\end{proof}

Next, we show that any interval of size $3T$ has conductance at least $T/4z$ for any minority set $S\subseteq V(\calG^{T,k,z})$.
\begin{lemma}\label{lem:clique-lower-bound-conductance}\leanok%
	Let $T,k,z \ge 1$ be integers with $z$ even. Then for any interval~$I\subseteq\N$ of length $3T$ and any set $S\subseteq V(\calG^{T,k,z})$ with $0<\abs{S}\leq n/2$, we have \[\sum_{t\in I} \phi^t(S) = \sum_{t \in I} \frac{e_t(S,\overline{S})}{\Vol(S)} \ge T/4z\,.\]
\end{lemma}
\begin{proof}
	Throughout, let $\calG = \calG^{T,k,z}$ and $n=kz=\abs{V(\calG)}$. Let $S \subseteq V(\calG)$ with $0<\abs{S} \le n/2$. Since $\calG$ is regular, this implies $0<\Vol(S) \le \tfrac12\Vol(V(\calG))$.
	We split into two cases depending on $S$.

	\medskip\noindent\textbf{Case 1:} For all $x \in [z]$, we have $\abs{S \cap V_x} \le k/2$. In this case, let $t \in I$ be arbitrary, and let $C_1,\dots,C_r$ be the cliques present in $\calG$ at time $t$. Then for all $i\in [r]$, we have $\abs{S \cap V(C_i)} \le k$. Thus by \cref{lem:static-clique-conductance},
	\begin{align*}
		\phi^t(S)=\frac{e_t(S,\overline{S})}{\Vol(S)}
		 & = \frac{1}{\Vol(S)}\sum_{i=1}^r e(V_x \cap V(C_i), V(C_i) \setminus S) \\
		 & \ge \frac{1}{\Vol(S)}\sum_{i=1}^r \frac{\Vol(S \cap V(C_i))}{2}
		= \frac{1}{2}.
	\end{align*}
	Thus, since $I$ has length $3T$,
	\begin{equation}\label{eq:lower-bound-conductance-case-1}%
		\sum_{t \in I} \frac{e_t(S,\overline{S})}{\Vol(S)} \ge \frac{3T}{2} > \frac{T}{4z}.
	\end{equation}

	\medskip\noindent\textbf{Case 2:} There exists $x \in [z]$ with $\abs{S \cap V_x} > k/2$. Since $\sum_{y \in [z]}|S \cap V_y| = |S| \le kz/2$, there also exists $y \in [z]$ with $\abs{S \cap V_y} \le k/2$; hence there exists $\ell \in [z]$ with $\abs{S \cap V_\ell} \le k/2$ and $\abs{S \cap V_{\ell+1}} > k/2$ (modulo $z$).
Since $I$ is of length $3T$ and the intervals $I_j$ are of length $T$, $I$ must contain two successive intervals $I_j$ and $I_{j+1}$; throughout one of these intervals, say $I'$, the clique $C$ spanned by $V_\ell$ and $V_{\ell+1}$ is a component of $\calG$. Let $S' = S \cap V(C)$; then for all $t \in I'$, we have
	\begin{equation*}
		e_t(S,\overline{S}) \ge e_t(S', V(C) \setminus S').
	\end{equation*}
	Observe that since $\abs{S \cap V_\ell} \le k/2$ and $\abs{S \cap V_{\ell+1}} > k/2$, we have $k/2 \le \abs{S'} \le 3k/2$. Let $X$ be whichever of $S'$ and $V(C) \setminus S'$ has size at most $k$, so that $e_t(S', V(C) \setminus S') = e_t(X, V(C) \setminus X)$; then by \cref{lem:static-clique-conductance}, it follows that
	\[
		e_t(S,\overline{S}) \ge \Vol(X)/2.
	\]
	Again since $k/2 \le \abs{S'} \le 3k/2$, we have $\Vol(X) \ge k(2k-1)/2$, and $\Vol(S)$ is at most $\Vol(V(\calG)) = n(2k-1)$; thus
	\begin{equation}\label{eq:lower-bound-conductance-case-2}%
		\sum_{t\in I}\frac{e_t(S,\ov{S})}{\Vol(S)} \ge \sum_{t \in I'}\frac{k(2k-1)/4}{n(2k-1)} \ge \frac{Tk}{4n} = \frac{T}{4z}.
    \end{equation}

	The result now follows immediately on combining \cref{eq:lower-bound-conductance-case-1} from Case 1 and \cref{eq:lower-bound-conductance-case-2} from Case 2.
\end{proof}

Our next goal is to bound the likely absorption time of the voter model on $\calG^{T,k,z}$ below in \cref{lem:lower-bound-absorption}. To do so, we first re-prove a (well-known) upper bound on absorption time for the clique; this will allow us to show that within an interval $I_j$ with $j$ odd (say), in every clique component of $G_1$, the voter model reaches consensus before the end of $I_j$ with high probability.
\begin{lemma}\label{lem:static-clique-voter}\leanok%
	There exists a positive integer $\Gamma$ such that, for all positive integers~$k$ and~$\alpha$, the following holds: From any initial state, with probability at least $1 - 2^{-\alpha}$, the standard voter model on the (static) clique on $2k$ vertices reaches consensus within $\Gamma\alpha k$ time.
\end{lemma}
\begin{proof}
	By \cref{lem:static-clique-conductance}, the clique $H$ on $2k$ vertices has conductance at least $1/2$. It follows, for example, from \cite[Theorem~1.1(i)]{DBLP:conf/icalp/BerenbrinkGKM16} that there exists $\Gamma'\geq1$ such that with probability at least $1/2$, from any initial state, the standard voter model on $H$ absorbs within time at most $\Gamma'\abs{V(H)}/\phi(H) \leq 4\Gamma' k\leq \Gamma k$. We apply this fact~$\alpha$ times to obtain the result.
\end{proof}

\begin{lemma}\label{lem:lower-bound-absorption}\leanok%
	There exists $\Gamma\geq1$ such that the following holds. Let $T,k,z \ge 1$ be integers with $z$ even and $z \ge 20$, and suppose $T \ge 10\Gamma k\log z$. Then there exists $s_0 \subseteq V(\calG^{T,k,z})$ such that with probability at least $1/2$, the standard voter model on $\calG^{T,k,z}$ with initial minority set $s_0$ does not absorb within time $Tz^2/128$.
\end{lemma}
\begin{proof}
	Throughout, let $\calG = \calG^{T,k,z}$. Let $\Gamma$ be the constant of \cref{lem:static-clique-voter}. Let $s_0 = V_1 \cup \dots \cup V_{z/2}$. In the voter model $(A_t)$ from initial state $s_0$, let $\tau$ be the earliest time $t$ divisible by $T$ at which $A_t$ is not precisely a set of the form $V_i \cup V_{i+1} \cup \dots \cup V_j$ (modulo $z$); thus $\tau$ is the earliest start of an interval $I_j$ in which the opinions do not form two contiguous intervals in $V_1,\dots,V_z$. For all $i \ge 0$, define
	\[
		W_i = \begin{cases}
		\abs{\{j \in [z]\colon V_j \subseteq A_{iT}\}} & \mbox{ if $iT < \tau$},\\
        W_{i-1} & \mbox{ otherwise.}
        \end{cases}
        \]

	Thus $W_0 = z/2$, and if $iT < \tau$ and $W_i \notin \{0,z\}$ then $A_{iT} \notin \{\emptyset, V(\calG)\}$ (and so the voter model has not absorbed by time $iT$).

	Our proof will now proceed in two parts. Let $\ell = \lceil z^2/128 \rceil$. First, we claim that $(W_i)$ is a martingale, and so by applying Azuma's inequality it is likely that $W_\ell \notin \{0,z\}$. Second, we argue it is likely that $\tau > \ell T$. Put together with a union bound, the two claims will immediately imply the result.

	\medskip\noindent\textbf{Claim 1: With probability at least $7/10$, $W_\ell \notin \{0,z\}$.} Let $\calM_i = (A_0, A_T,\allowbreak \dots, A_{iT})$. We first establish that $(W_i)$ is a martingale with filtration $\sigma(\calM_i)$. Fix $i \ge 0$, and let $M_i$ be a possible value of $\calM_i$. If $\calM_i = M_i$ determines $iT \ge \tau$, then $\E(W_{i+1} \mid \calM_i = M_i) = W_i$ as required, so suppose not. Let $a_{iT}$ be the value of $A_{iT}$ determined by $M_i$. Throughout the interval $I_{i+1}$, $\calG$ is constant and either one or two cliques in $\calG$ contain vertices of both opinions. Moreover, in each such clique, exactly half of the vertices have each opinion in $a_{iT}$. Thus by symmetry, each opinion is equally likely to dominate in each clique within the interval $I_{i+1}$. It follows that
	\begin{align*}
		\abs{W_{i+1} - W_i}                       & \le 2,                                                \\
		\pr(W_{i+1}=W_i+1 \mid \calM_i = M_i) & = \pr(W_{i+1}=W_i-1 \mid \calM_i = M_i), \mbox{ and } \\
		\pr(W_{i+1}=W_i+2 \mid \calM_i = M_i) & = \pr(W_{i+1}=W_i-2 \mid \calM_i = M_i).
	\end{align*}
	Thus once again, $\E(W_{i+1} \mid \calM_i = M_i) = W_i$. We conclude that $(W_i)$ is indeed a martingale with filtration $\sigma(\calM_i)$.

	We now apply Azuma's inequality (\cref{lem:azuma-inequality}) to $(W_i)$, taking $c_i=2$ for all $i$ and taking the $T$ of \cref{lem:azuma-inequality} to be our $z/2$. Since $z \ge 20$, $\ell = \lceil z^2/128 \rceil \le z^2/64$; thus we obtain
\begin{equation}\label{eq:lower-bound-absorb}%
		\pr(W_\ell \in \{0,z\}) = \pr(\abs{W_\ell - W_0} \ge z/2) \le 2e^{-(z/2)^2/8\ell} \le 2/e^2 < 3/10.
	\end{equation}

	\medskip\noindent\textbf{Claim 2: With probability at least $4/5$, $\tau > \ell T$.} We will apply a union bound over all intervals $I_j$ with $j \in [\ell]$. To this end, fix $i \ge 0$, let $M_i$ be a possible value of $\calM_i$, and suppose $M_i$ determines $iT < \tau$. Then as in Claim 1, throughout the interval $I_{i+1}$, $\calG$ is constant and either one or two cliques in $\calG$ contain vertices of both opinions. By \cref{lem:static-clique-voter} applied with $\alpha = \lceil 3\log_2 z\rceil < 10\log z$, each clique absorbs within $I_{i+1}$ with probability at least $1 - 2^{-\alpha} \ge 1 - 1/z^3$. Taking a union bound over all (one or two) cliques in $I_{i+1}$ with vertices of both opinions, we see that $\pr(\tau > (i+1)T \mid \calM_i=M_i) \ge 1 - 2/z^3$. Taking a union bound over $i \in [\ell]$, it follows that
\begin{equation}\label{eq:lower-bound-absorb-2}%
		\pr(\tau \le \ell T) \le 2\ell/z^3 < 1/5.
	\end{equation}

	The result now follows from a union bound over \cref{eq:lower-bound-absorb} and \cref{eq:lower-bound-absorb-2}, together with the fact that if $\ell T < \tau$ and $W_\ell \notin \{0,z\}$ then $A_{\ell T} \notin \{\emptyset, V(\calG)\}$:
    \[
        \pr(A_{\ell T} \in \{\emptyset, V(\calG)\}) \le \pr(\ell T \ge \tau) + \pr(W_\ell \in \{0,z\}) \le 1/2.
    \]
\end{proof}

We are now ready to prove the main result of the section.

\begin{theorem}\label{thm:lower-bound-intervals}\leanok%
	There exists a constant $c>0$ such that the following holds. Let $d(x)\colon\N\rightarrow\N$ such that, for all $x \in \N$, $1\leq d(x)\leq  \sqrt{x}/\log x$ and $d(x)$ is odd. Then there exist functions $\phi\colon \N\to\R$ and $\Delta\colon\N\to\N$ such that $\phi(x) = \Theta(d(x)/x)$ and $\Delta(x) = \Theta(x/d(x))$ as $x \to \infty$ and the following holds. For infinitely many values of $x$, there exist an $n$-vertex temporal graph $\calG$ and $s_0 \subseteq V(\calG)$ with the following properties.
    \begin{enumerate}[(i)]
        \item $\calG$ is $d(x)$-regular and $x \le n \le 2x$.
        \item For all time intervals $I$ of length $\Delta(x)$ and all sets $S \subseteq V(\calG)$ with $0 < \Vol(S) \le \Vol(V(\calG))/2$, $\tfrac{1}{\abs{I}}\sum_{t \in I} \phi^t(S) \ge \phi(x)$. 
        \item There exists $s_0 \subseteq V(\calG)$ such that with probability at least $1/2$, the standard voter model on $\calG$ with initial minority set $s_0$ does not absorb within time $cx\Delta(x)/(d(x)\phi(x))$.
    \end{enumerate}
\end{theorem}

\begin{proof}
	Let $x$ be suitably large (to be determined later), let $k \coloneqq (d(x)+1)/2$, let $z \coloneqq \lceil x/k\rceil$, and let $n\coloneqq kz$. Observe that $x \le n \le 2x$ as required by (i) since $d(x) \le \sqrt{x}/\log x$ and $x$ is large. Let $T\coloneqq 4z$. Let $\calG \coloneqq \calG^{T,k,z}$, and observe that $\calG$ is a $d(n)$-regular graph as required by (i).

	Let $\phi(x)\coloneqq \frac{1}{12z}$ and $\Delta(x)\coloneqq 3T$, observing that $\phi(x) = \Theta(d(x)/x)$ and $\Delta(x) = \Theta(x/d(x))$ as $x\to\infty$ as required in the statement. By \cref{lem:clique-lower-bound-conductance}, for all half-volume sets $S\subseteq V(\calG)$,
	\begin{equation*}
		\frac{1}{3T}\sum_{t\in I_i}\phi^t(S)\geq 1/(12z)=\phi(x),
	\end{equation*}
    as required by (ii).

	Let $\Gamma>0$ be as in \cref{lem:lower-bound-absorption}. Observe that since $n = \Theta(x)$, $k = \Theta(d(x))$, and $d(x) = o(\sqrt{x/\log x})$, we have $z = n/k = \omega(\sqrt{x\log x})$; thus since $x$ is sufficiently large, we have $z \ge 20$ and $T = 4z \ge 20\Gamma k\log x \ge 10\Gamma k\log z$ as required by \cref{lem:lower-bound-absorption}. Hence, \cref{lem:lower-bound-absorption} applies to $\calG$; let $s_0 \subseteq V(\calG)$ be the minority set from the lemma statement, and observe that with probability at least $1/2$, the standard voter model on $\calG$ with initial state $s_0$ does not absorb within time $Tz^2/128$.
	It remains to prove that $Tz^2/128 \ge cn\Delta(x)/(d(x)\phi(x))$ for a suitable choice of $c$. Since $\phi(x)=1/(12z)$ we have that $z=1/(12\phi(x))$, and $T = \Delta(x)/3$; thus $Tz^2/128 \geq \Delta(x)z/(10^4\phi(x))$. Moreover, $z=n/k\geq x/d(x)$, so it follows that
	   \[
		\frac{Tz^2}{128} \ge \frac{\Delta(x) x}{10^4\phi(x)d(x)}.
        \]
	On taking $c=1/10^4$, (iii) follows, and thereby so does the result.
\end{proof}

\LowerBoundf*
\begin{proof}
    We apply \cref{thm:lower-bound-intervals}, taking $d(x) = d$ for all $x$. Take $c'$, $\phi'$ and $\Delta'$ to be the $c$, $\phi$ and $\Delta$ of \cref{thm:lower-bound-intervals}, let $x$ be suitably large, let $\calG$ be a temporal graph satisfying \cref{thm:lower-bound-intervals}(i)--(iii) for $x$, and let $n\coloneqq|V(\calG)|$ so that $x \le n \le 2x$ and in particular $n$ is suitably large. We define $\phi(n) \coloneqq \phi'(x)$, $\Delta(n) \coloneqq \Delta'(x)$, and $c \coloneqq 2c'$.
    
    Observe that since $n = \Theta(x)$, $\phi'(x) = \Theta(d(x)/x)$ and $\Delta'(x) = \Theta(x/d(x))$, $\phi$ and $\Delta$ have the correct growth rates. Moreover: $\calG$ is $d$-regular by \cref{thm:lower-bound-intervals}(i); (i) of the statement is immediate from \cref{thm:lower-bound-intervals}(ii); and (ii) of the statement follows from \cref{thm:lower-bound-intervals}(iii) together with $x \ge n/2$ and $c=2c'$.
\end{proof}

\LowerBounds*

\begin{proof}
	Let $\phi(n)$ and $\Delta(n)$ be as in \cref{thm:lower-bound-odd-regular-averaged}, then by the definition of conductance it holds that $\Phi(\calG)\geq \phi(n)/\Delta(n)$, therefore

	\[
		cn\Delta/(d(n)\phi(n))\geq cn/(d(n)\Phi(\calG))
	\]
	The result follows by choosing $d=3$, and $C=c/3$.
\end{proof}

%% file: sections/5_ai.tex
\section{AI Statement}

The authors wrote and checked every line of the paper by hand, and assume responsibility for all content. LLMs were used for the Lean formalisation (with author verification of the main theorems' statements and dependencies), additional correctness checking based on the Lean formalisation, and as a form of search engine for minor questions. The original proofs of our results were entirely human-generated, predating the Lean formalisation; the current contents of \cref{sec:case-2} are based on an AI-suggested simplification.